\documentclass[11pt]{article}
\usepackage[T1]{fontenc}
\usepackage[UKenglish]{isodate} 
\usepackage{comment,xfrac}
\usepackage{import}
\usepackage{url}

\usepackage{xcolor}
\usepackage{pdflscape}
\usepackage{placeins}
\usepackage{subcaption}
\usepackage{bbm}
\usepackage{listings}
\usepackage{color}
\usepackage{amssymb}
\usepackage{graphicx}
\usepackage{tikz}
\usetikzlibrary{shapes, shapes.arrows}
\usepackage{amsmath}
\usepackage{lipsum}
\usepackage[a-1b]{pdfx}

\usepackage{hyperref}  
			
\hypersetup{colorlinks, citecolor=black, filecolor=black, linkcolor=black, urlcolor=black}
\usepackage{bm} 

\usepackage{glossaries} 
\newacronym{ap}{AP}{Access Point}
\newacronym{ac}{AC}{Access Category}
\newacronym{ack}{ack}{acknowledgement}
\newacronym{acs}{ACs}{Access Categories}
\newacronym{arq}{ARQ}{Automatic Repeat Request}
\newacronym{arf}{ARF}{Automatic Rate Fallback}
\newacronym{aarf}{AARF}{Adaptive ARF}
\newacronym{aifs}{AIFS}{Arbitrary Inter Frame Spacing}
\newglossaryentry{aloha}{name={ALOHA},description={An early and simple multiple access protocol.}}
\newacronym{awgn}{AWGN}{Additive White Gaussian Noise}

\newacronym{bats}{BATS}{BATched Sparse}
\newacronym{ber}{BER}{Bit Error Rate}
\newglossaryentry{bluetooth}{name={Bluetooth}, description={A short range
		Personal Area Network (PAN) standard}}
\newglossaryentry{bp}{name={belief propagation},description={A method for decoding Luby Transform (LT) coded messages}}
\newacronym{bs}{BS}{Base Station}

\newacronym{cav}{CAV}{Connected and Autonomous Vehicle}
\newacronym{cara}{CARA}{Collision Aware Rate Adaptation}
\newglossaryentry{cooperative spectrum sensing}{name={cooperative spectrum sensing}, description={A system in which cognitive radio users collaborate to achive more accurate spectrum sensing}}
\newacronym{cts}{CTS}{Clear to Send}
\newacronym{cca}{CCA}{Clear Channel Assessment}
\newacronym{cdma}{CDMA}{Code Division Multiple Access}
\newacronym{csma/ca}{CSMA/CA}{Carrier Sense Multiple Access/Collision Avoidance}
\newacronym{csma}{CSMA}{Carrier Sense Multiple Access}
\newacronym{crc}{CRC}{Cyclic Redundancy Check}
\newacronym{ca}{CA}{Collision Avoidance}
\newglossaryentry{cognitive radio}{name={cognitive radio},description={A radio system which can adapt its operational parameters in order to improve performance, and in particular to exploit whitespace as a secondary spectrum user}}
\newacronym{cw}{CW}{Contention Window}

\newacronym{darpa}{DARPA}{Defense Advanced Research Projects Agency}
\newacronym{d2d}{D2D}{Device to Device}
\newacronym{doa}{DOA}{Direction of Arrival}
\newacronym{dcf}{DCF}{Distributed Coordination Function}
\newacronym{difs}{DIFS}{DCF Inter Frame Space}

\newacronym{eirp}{EIRP}{Effective Isotropic Radiated Power}
\newacronym{edca}{EDCA}{Enhanced Distributed Channel Access}
\newacronym{ess}{ESS}{Evolutionary Stable Strategy}

\newacronym{fec}{FEC}{Forward Error Correction}
\newacronym{fdd}{FDD}{Frequency Division Duplex}
\newacronym{fft}{FFT}{Fast Fourier Transform}

\newglossaryentry{gps}{name={GPS},description={Global Positioning System}}

\newacronym{harq}{HARQ}{Hybrid Automatic Repeat Request}
\newacronym{hcca}{HCCA}{HCF Controlled Channel Access}
\newacronym{hcf}{HCF}{Hybrid Coordination Function}
\newacronym{hcf/edca}{HCF/EDCA}{Hybrid Coordination Function/Enhanced Distributed Channel Access}

\newacronym{iot}{IoT}{Internet of Things}
\newglossaryentry{iid}{name={i.i.d}, description={Independently and Identically Distributed}}
\newacronym{itu}{ITU}{International Telecommunication Union}
\newacronym{isp}{ISP}{Internet Service Provider}

\newglossaryentry{jamming}{name={jamming},description={Covertly transmitting on the same frequency as another system, in order to intentionally cause interference and disrupt its operation}}

\newacronym{lan}{LAN}{Local Area Network}{
\newacronym{lt}{LT}{Luby Transform}

\newacronym{mod}{MOD}{Ministry Of Defence}
\newacronym{mimo}{MIMO}{Mulitple Input Multiple Output}
\newacronym{mpdu}{MPDU}{MAC Packet Data Unit}
\newacronym{mac}{MAC}{Medium Access Control}

\newacronym{nack}{nack}{negative acknowledgement}
\newacronym{ns3}{NS3}{Network Simulator 3}
\newacronym{nav}{NAV}{Network Allocation Vector}

\newacronym{ofdm}{OFDM}{Orthogonal Frequency Division Multiplexing}

\newacronym{phy}{PHY}{Physical Layer}
\newacronym{per}{PER}{Packet Error Rate}
\newacronym{pcf}{PCF}{Point Coordination Function}
\newacronym{pc}{PC}{Point Coordinator}
\newacronym{ppbp}{PPBP}{Poisson Pareto Burst Process}
\newacronym{prng}{PRNG}{Pseudo-Random Number Generator}
\newacronym{psd}{PSD}{Power Spectral Density}

\newacronym{qos}{QoS}{Quality of Service}
\newacronym{qoe}{QoE}{Quality of Experience}

\newacronym{rca}{RCA}{Rate Control Algorithm}
\newacronym{rf}{RF}{Radio Frequency}
\newacronym{rlnc}{RLNC}{Random Linear Network Coding}
\newacronym{rts}{RTS}{Request to Send}
\newacronym{rx}{RX}{Receive}
\newacronym{rtscts}{RTS/CTS}{Request to Send/Clear to Send}
\newacronym{rraa}{RRAA}{Robust Rate Adaptation Algorithm}
\newacronym{rcpi}{RCPI}{Received Channel Power Indicator}

\newacronym{shre}{SHRE}{Single Hop Range Extension}
\newacronym{sifs}{SIFS}{Short Inter Frame Space}
\newacronym{sinr}{SINR}{Signal to Interference plus Noise Ratio}
\newacronym{snr}{SNR}{Signal to Noise Ratio}
\newglossaryentry{son}{name={self organising network},description={a system in which network nodes operate independently, with limited knowledge of the actions and states of only its immediate neighbours, whilst remaining agile in response to environmental change}}
\newacronym{ssid}{SSID}{Service Set IDentifier}
\newglossaryentry{spectrum sensing}{name={spectrum sensing}, description={An operation to detect whitespace, and the presence of primary users}}
\newglossaryentry{spectrum monitoring}{name={spectrum monitoring}, description={An simple, cheap form of spectrum sensing, to be performed during transmission}}
\newacronym{sta}{STA}{Station}
\newacronym{stnc}{STNC}{Space-Time Network Coding}
\newacronym{sd}{SD}{Standard Definition}

\newacronym{tact}{TACT}{transfer actor-critic}
\newglossaryentry{temporal relevance}{name={temporal relevance}, description={The amount of time for
which data will remain \emph{relevant} to its intended recipient. Data streams in which messages are
quickly superseded by newer messages have low temporal relevance}}
\newacronym{tdd}{TDD}{Time Division Duplex}
\newglossaryentry{timestep}{name={time-step}, description={An interval of discrete time}}
\newacronym{txop}{TXOP}{Transmit Opportunity}
\newacronym{tx}{TX}{Transmit}
\newacronym{tcp}{TCP}{Transmission Control Protocol}
\newacronym{tdma}{TDMA}{Time Division Multiple Access}

\newacronym{v2v}{V2V}{Vehicle to Vehicle}
\newacronym{v2i}{V2I}{Vehicle to Infrastructure}

\newglossaryentry{wifi}{name={Wi-Fi},description={A popular wireless radio standard}}
\newacronym{wlan}{WLAN}{Wireless Local Area Network}
\newglossaryentry{wimax}{name={WiMAX},description={A radio standard for metropolitan area networks}}
\newglossaryentry{whitespace}{name={whitespace},description={A period of time for which licensed spectrum is unoccupied by its primary user}}

\usepackage{amsthm}
\newtheorem{lemma}{Lemma}
\newtheorem{theorem}{Theorem}
\newtheorem{cor}{Corollary}
\newtheorem{prop}{Proposition}
\theoremstyle{definition}

\theoremstyle{remark}
\newtheorem*{remark}{Remark}

\newcommand{\matr}[1]{\bm{#1}}
\renewcommand{\vec}[1]{\boldsymbol{#1}}

\newcommand{\prob}{{\mathbb P}}
\newcommand{\E}{{\mathbb E}}
\newcommand{\ER}{{Erd\H{o}s-R\'enyi }}

\newcommand{\nats}{{\mathbb N}}

\newcommand{\kl}{D}
\newcommand{\gf}{{\mathbb{F}}}
\newcommand{\boldx}{\boldsymbol{x}}
\newcommand{\boldzero}{\boldsymbol{0}}

\author{Mark~A.~Graham\thanks{Roke Manor Research. Research supported by EPSRC grant EP/I028153/1, and Roke Manor Research.}, Ayalvadi~J.~~Ganesh\thanks{School of Mathematics, University of Bristol, Bristol, UK.}, Robert~J.~Piechocki\thanks{Dept. of Electronic Engg., University of Bristol, Bristol, UK.}}

\title{Low latency allcast over broadcast erasure channels}

\begin{document}

\maketitle

\begin{abstract}
Consider $n$ nodes communicating over an unreliable broadcast channel. Each node has a single packet that needs to be communicated to all other nodes. Time is slotted, and a time slot is long enough for each node to broadcast one packet. Each broadcast reaches a random subset of nodes. The objective is to minimise the time until all nodes have received all packets. We study two schemes, (i) random relaying, and (ii) random linear network coding, and analyse their performance in an asymptotic regime in which $n$ tends to infinity. Simulation results for a wide range of $n$ are also presented.
\end{abstract}
\section{Introduction}
The problem of multicasting a message from a single source to multiple receivers over wired networks such as the Internet has been widely studied. Content distribution networks and live streaming were some of the major motivations. Early work sought centralised algorithms for building optimal multicast trees according to specified performance measures; see, e.g., ~\cite{waxman, kompella93}. This often gave rise to NP-hard optimisation problems, typically variants of the Steiner tree problem, and much work was devoted to developing approximation algorithms for these problems. Subsequently, interest shifted to distributed algorithms for constructing multicast trees on overlay networks~\cite{castro03, lee06}. These works mainly focused on throughput and on reliability under erasures, e.g., due to packet losses caused by congestion. A lightweight gossip-style scheme for content dissemination in peer-to-peer networks was presented in~\cite{sanghavi}. A number of works proposed network coding as a mechanism to maximise throughput~\cite{zli05, deb06, ho06}. The delay performance of network coding was analysed in~\cite{medard06, medard08}.

The increasing prevalence of wireless communication at the network edge motivates interest in broadcast channels as opposed to point-to-point links. Live streaming of a single source in dense wireless networks has been studied in~\cite{costa14}. A motivating application for the work in this paper comes from autonomous vehicles, which need to exchange data on velocity, acceleration and manoeuvring intentions reliably and with very low latency. The work presented here is also relevant to flooding of routing information in wireless ad-hoc networks, and to cyber-physical systems composed of large numbers of sensors and control devices communicating over wireless channels. Common features of these applications are short messages, originating from many nodes in the network, and needing to be communicated to many other nodes with very low latency. Moreover, communication links are unreliable and one-hop communication is not possible between all pairs of nodes. These are the features addressed in the work presented here. In addition, as low latency is achieved by minimising the number of transmissions, it has the side-effect of reducing energy use. Hence, the algorithms presented here are also relevant to applications in which energy is the binding constraint~\cite{fragouli08}. 

The remainder of this section places our contribution within the context of related work. The system model and problem statement are presented in Section~\ref{sec:model}, followed by a lower bound on the best achievable delay in Section~\ref{sec:lower_bd}. Two specific algorithms, based on random forwarding and random linear network coding respectively, are presented in Sections~\ref{sec:fwd} and~\ref{sec:coding}, along with an analysis of their performance. We present simulation results in Section~\ref{sec:sim}, and conclude in Section~\ref{sec:concl}.

\subsection{Related Work}

Consider a population of $n$ agents represented by the nodes of a directed communication graph $G=(V,E)$. Time is discrete and a time slot is long enough for a node to transmit one message of a fixed size, known as a packet. A transmission by node $u$ in time slot $t$ is received correctly by $v$ if $(u,v)\in E$ and no other node $w$ such that $(w,v)\in E$ transmits in the same time slot. If two or more such in-neighbours of $v$ transmit in the same time slot, then their packets collide and are erased at $v$ (but may be received correctly at other nodes). There are no errors or erasures other than due to collisions. 

Two problems have been studied on this system model. In the \emph{broadcasting} problem, a single node $v$ has a single packet which needs to be communicated to all other nodes. In the \emph{gossiping} or \emph{allcast} problem, each node has a single packet, and all nodes need to receive all packets. In each time slot, nodes may retransmit any packet they have previously received, but may not alter the contents of a packet. The objective is to construct a schedule of minimal length which ensures that all nodes receive the required packets. Both centralised and distributed versions of these problems have been studied, as well as versions in which the graph $G$ is known or unknown to the agents.

It was shown in~\cite{chlamtac85} that even the broadcasting problem is NP-hard. An approximation algorithm which achieves a cost within a logarithmic (in $n$) factor of the optimal was presented for the broadcasting problem in~\cite{baryehuda92}, and extended to the gossiping problem in~\cite{baryehuda93}. The approach taken in both these works is to construct breadth-first-search (BFS) trees, which incurs a one-time cost. Other approaches to the broadcasting and gossiping problems, which yield polylogarithmic approximations, are presented in~\cite{chlebus01, gaspot02, chrobak04}. The above works address general networks, which might be known or unknown to the agents. 

The special case of random geometric graphs, also known as the Gilbert disk or Boolean model, was studied in~\cite{ravishankar95, gupta99, huang08}. A constant factor approximation algorithm was obtained for such graphs in~\cite{huang08} using BFS trees. An asymptotically optimal algorithm for the one-dimensional case was presented in~\cite{ravishankar95}. Our paper complements these works by studying the special case of dense \ER random graphs. We show that random forwarding achieves a logarithmic approximation factor while respecting the constraint of only retransmitting received packets, while network coding achieves a constant factor approximation but violates this constraint.

The works cited above address collisions as well as the choice of which packet to transmit in each time slot in their scheduling algorithms. We abstract out collisions, assuming orthogonal channels. (Alternatively, collisions can be dealt with at a lower layer in the protocol stack, incurring an $O(\log n)$ penalty on average.) This is the same system model as considered in~\cite{firooz13}. They also use network coding, and obtain delay bounds for general networks. Specialising their results to our network model recovers a constant factor approximation guarantee, but with poorer (and less explicit) constants than we obtain. In addition, our proof techniques are very different, being based on concentration inequalities rather than combinatorial identities.

Finally, a complete graph with random edge capacities is studied in~\cite{swamy13}. The authors characterise the capacity region, and present push-pull algorithms (which do not employ network coding) for the gossiping problem which they show to be asymptotically optimal.

\section{System Model}\label{sec:model}
We consider a set of $n$ agents or nodes, each of whom has a single packet to transmit. All packets are of the same size. Time is discrete, and divided into slots or rounds. Each agent gets to broadcast a single packet in each round. Each broadcast is received error-free at a subset of agents, and erased completely at the others; no packets are received partially or with errors. We ignore issues related to contention, assuming that they have been dealt with at a lower layer. We assume that agents have unlimited storage and computational power. The algorithms that we consider require memory growing linearly with the number of agents, and computational power growing polynomially (at most as a cubic).

Let $G_t=(V,E_t)$ denote the directed graph whose vertices are the agents; if edge $(u,v)$ is in $E_t$, it denotes that the broadcast by agent $u$ in time slot $t$ is received by agent $v$. The objective is to minimise the time until all packets  reach all agents. Agents may cooperate to achieve this objective, but subject to the following restrictions. The size of the packet transmitted by an agent in any round must be a constant (i.e., not growing with $n$), and its content must only depend on the agent's original packet, and the packets it has received in previous rounds. 
We assume that agents are unaware of the graphs, $G_t$. In fact, the algorithms we propose only require packets to have identifiers; it is not necessary for agents to have identifiers, or to know who they can communicate with.
Finally, we assume that no agents are faulty or malicious, i.e., that they all implement any specified algorithm complying with the restrictions.
We now state two assumptions about the graph sequence, $G_t$, $t\in \nats$. 

\noindent{\bf Assumptions}
\begin{enumerate}
    \item[A1.]
    In each time slot $t$, the communication graph $G_t$
is a directed \ER random graph $G(n,p)$, for a fixed $p>0$. In other words, each directed edge is present with probability $p$, independent of all other edges.
    \item[A2.]
    The graph $G_t$ is fixed once and for all, i.e., $G_t=G_1$ for all $t\in \nats$.
\end{enumerate}

Assumption A1 says that all nodes face a similar environment, and hence have approximately the same in- and out-degrees as each other, concentrated sharply around the mean value of $np$. Furthermore, there is no spatial structure, in physical or latent space. The parameter $p\in (0,1]$ determines the edge density. The assumption is plausible when agents are in fairly close proximity, and connectivity is dominated by random fading rather than by attenuation. Examples include autonomous vehicles within a small neighbourhood or stretch of road, or sensors or IoT devices within a small building.

Assumption A2 is appropriate for environments with slow fading, compared to the time scale required for the allcast to complete; as our focus is on low-latency applications, the assumption seems realistic. In addition, intuition suggests, and the simulations reported in Section~\ref{sec:sim} confirm, that if the random graphs evolve over time, then that only reduces the dissemination time of the algorithms we study. As such, the fixed topology setting mandated by this assumption appears to be a worst-case scenario for the allcast problem.

\section{A lower bound} \label{sec:lower_bd}

In this section, we present an elementary lower bound on the number of rounds required by any allcast algorithm. We say that a property holds ``with high probability (whp)" if the probability that the property is satisfied tends to 1 as $n$ tends to infinity.

\begin{prop} \label{prop:allcast_lowerbd}
Let $G_t=(V,E_t)$, $t=1,2,3,\ldots$ be a sequence of directed communication graphs with common vertex set, $V$. Let $d^{\rm in}_t(v)$ denote the in-degree of node $v$ in $G_t$, and let 
$$
\tau_v = \inf \{ T: \sum_{t=1}^T d^{\rm in}_t(v) \geq n-1 \} .
$$
Then, for any allcast algorithm, the time until all nodes have all packets is at least as large as $\max_{v\in V} \tau_v$. 

In particular, if $G_t \equiv G$ for all $t\in \nats$ and $d^{\rm in}_{\min}$ is the minimum in-degree of any vertex in $G$, then allcast requires at least $(n-1)/d^{\rm in}_{\min}$ rounds, for any algorithm.
\end{prop}

\begin{proof} 
The proof is straightforward. Each node needs to receive $n-1$ distinct packets in total, and node $v$ can receive at most $d^{\rm in}_t(v)$ distinct packets in round $t$. Hence, the earliest time by which node $v$ can receive $n-1$ distinct packets is $\tau_v$, and the claim follows. \end{proof}

The in-degree of vertices of directed \ER random graphs $G(n,p)$ concentrate sharply around $np$, as $n$ tends to infinity with $p$ fixed. This yields the following corollary.

\begin{cor} \label{cor:lower_bd}
Suppose that for each $t\in \nats$, $G_t$ is an \ER random graph, $G(n,p)$. Given an allcast algorithm $\mathcal{A}$, let $T^{\rm all}(\mathcal{A})$ denote the random time until all nodes have all packets. Then, for any $q>p$, 
$$
\prob \Bigl( T^{\rm all}(\mathcal{A}) \leq \frac{1}{q} \Bigr) \leq \frac{1}{q} e^{-n(n-1)H(q;p)},
$$
where $H(q;p)=q\log\frac{q}{p}+(1-q)\log\frac{1-q}{1-p}$ denotes the relative entropy or KL-divergence of the  Bernoulli distribution with parameter $q$, denoted $Bern(q)$, with respect to the $Bern(p)$ distribution.
\end{cor}

\emph{Remark.} Note that we make no assumption about the joint distribution of the random graphs $G_t$, or about the algorithm $\mathcal{A}$. In particular, the algorithm may use knowledge of the entire sequence $G_t, t\in \nats$, and may even design the sequence, subject only to the constraint on the marginal distributions.

The proof uses a well-known bound on large deviations for binomial random variables, variously known as Bernstein's or Chernoff's inequality; it also follows readily from Sanov's theorem. We state it below for easy reference as we shall make repeated use of it. 

\begin{lemma}\label{lem:binom_ldp}
Suppose that $X$ is a binomially distributed random variable with parameters $(n,p)$, which we denote by $X\sim Bin(n,p)$. Then,
\begin{align*}
\prob(X>nq) \le \exp\big(-nH(q;p)\big), & \quad \forall \; q>p, \\
\prob(X<nq) \le \exp\big(-nH(q;p)\big), & \quad \forall \; q<p.
\end{align*}
\end{lemma}

\emph{Proof of Corollary~\ref{cor:lower_bd}.} 
By Proposition~\ref{prop:allcast_lowerbd}, for any algorithm $\mathcal{A}$, $T^{\rm all}(\mathcal{A}) \geq \tau_v$ for all $v\in V$. Hence,
$$
\sum_{v\in V} \sum_{t=1}^{T^{\rm all}(\mathcal{A})} d^{\rm in}_t(v) \geq n(n-1).
$$
Therefore, for the event $T^{\rm all}(\mathcal{A})\leq \frac{1}{q}$ to occur, there must be at least one $t\in \{ 1,\ldots,\lfloor \frac{1}{q} \rfloor \}$ such that $\sum_{v\in V} d^{\rm in}_t(v) \geq qn(n-1).$ Hence, by the union bound, 
$$
\prob \Bigl( T^{\rm all}(\mathcal{A})\leq \frac{1}{q} \Bigr) \leq \sum_{t=1}^{\lfloor 1/q \rfloor} \prob \Bigl( \sum_{v\in V} d^{\rm in}_t(v) \geq qn(n-1) \Bigr).
$$
But $\sum_{v\in V} d^{\rm in}_t(v)$ has a $Bin(n(n-1),p)$ distribution, since $G_t$ is a directed \ER random graph for each $t$. The claim of the corollary thus follows easily from Lemma~\ref{lem:binom_ldp}. \hfill $\Box$

\section{Random message forwarding} \label{sec:fwd} 
In this section, we propose and analyse a simple gossip-style algorithm for allcast. The analysis requires Assumption A2 to hold, but the algorithm can be implemented on time-varying networks. We first present two variants of the gossip algorithm.

{\bf Algorithm Relay}
\begin{enumerate}
    \item In the first round, each node broadcasts its own packet.
    \item Each node stores all packets received in previous rounds in memory, organised by the earliest round in which they were received.
    \item In each round $t\geq 2$, each node broadcasts a packet chosen uniformly at random from either 
    \begin{itemize}
        \item{Algorithm R1}: all packets it received in the first round, or 
        \item{Algorithm R2}: all packets it received in all previous rounds, $1,\ldots,t-1$.
    \end{itemize}
    In either case, the choices are mutually independent across nodes, and across rounds at the same node. 
\end{enumerate}

If Assumption A2 holds, the graph topology is fixed once and for all. As nodes only forward packets received in the first round in Algorithm R1, the algorithm can ensure all messages reach all nodes only if the graph diameter is at most two, i.e., only if there is a one or two-hop path between any pair of nodes. It is easy to show that the diameter of a dense \ER random graph $G(n,p)$ is indeed bounded above by 2 whp, for any $p>0$. This will be implicitly proved in the analysis we present. Algorithm R2 is more robust, and does not require the graph diameter to be bounded, but the best upper bound we can prove on its performance is not as good as that for Algorithm R1. We now state our performance guarantees.
\begin{theorem}\label{thm:baseline}
	Consider a sequence of systems indexed by the number of nodes, $n$, and satisfying
	Assumptions A1 and A2. For an allcast algorithm $R^{\prime}$, let $T^{\rm
	all}_n(R^{\prime})$  denote the number of rounds until every node has received every message
	in a system with $n$ nodes employing algorithm $R^{\prime}$. Let $\epsilon>0$ be fixed.
	Then,
	\begin{align*}
		\lim_{n\to \infty} \prob \big(T^{\rm all}_n(R1) >\tfrac{2(1+\epsilon)}{p}\log(n)\big) &= 0, \\
		\lim_{n\to \infty} \prob \big(T^{\rm all}_n(R2) >\tfrac{2(1+\epsilon)}{p^2}\log(n)\big) &= 0, 
	\end{align*}
	where the probabilities are calculated over randomness in both the communication graph and the algorithm.
\end{theorem} 
Fix $n\in \nats$, and let $G=(V,E)$ be a directed \ER random graph $G(n,p)$ denoting the
communication graph in the $n^{\rm th}$ system. For $v\in V$, denote its in- and out-neighbourhoods
in $G$ by $N^{\rm in}_v = \{ u\in V: (u,v)\in E \}$ and $N^{\rm out}_v = \{ w\in V: (v,w) \in E \}$,
and its in- and out-degrees by $d^{\rm in}_v = |N^{\rm in}_v|$ and $d^{\rm out}_v = |N^{\rm
out}_v|$. For $u,v \in V$, let $M_{uv}= N^{\rm out}_u \cap N^{\rm in}_v$ denote the set of nodes
which connect $u$ and $v$ via two-hop directed paths. 
	
	In proving the theorem, we will use the fact that the in- and out-degrees of all vertices are concentrated around their mean values, as are the cardinalities of the sets $M_{uv}$, $u.v\in V$. We now state and prove this fact. 
	
\begin{lemma} \label{lem:degree_conc}
Fix $\delta>0$, and let $\mathcal{E}^1_{\delta}$ and $\mathcal{E}^2_{\delta}$ denote the events, 
\begin{align*}
	\mathcal{E}^1_{\delta} &= \bigcap_{v\in V} \Bigl\{ \frac{d^{\rm in}_v}{np} \in (1-\delta,1+\delta) \mbox{ and } \frac{d^{\rm out}_v}{np} \in (1-\delta,1+\delta) \Bigr\}, \\
	\mathcal{E}^2_{\delta} &= \bigcap_{u,v\in V} \Bigl\{  \frac{|M_{uv}|}{(n-2)p^2} > 1-\delta \Bigr\}.
\end{align*}
	Then, 
\begin{equation*}
	\prob( (\mathcal{E}^1_{\delta})^c ) \leq 2n e^{-\gamma (n-1)}, \quad
	\prob( (\mathcal{E}^2_{\delta})^c ) \leq n^2 e^{-\gamma (n-2)}, 
\end{equation*}
where $\gamma = \min \{ \kl((1-\delta)p;p); \kl((1+\delta)p;p) \} >0$.
The dependence of $\gamma$ on $\delta$ has not been made explicit in the notation. 
\end{lemma}

\begin{proof}
	Observe that the in- or out-degree of a given node in a given timestep is a $Bin(n-1,p)$ random variable. Similarly, 
    $|M_{uv}|$ is a $Bin(n-2,p^2)$ random variable, as each $w\neq u,v$ is in the set $M_{uv}$ with probability $p^2$, independent of all other nodes.
    Hence, the claims of the lemma follow from Lemma~\ref{lem:binom_ldp} and the union bound.
\end{proof}

\begin{proof}[Proof of Theorem \ref{thm:baseline}]
	Fix $u, v \in V$. For $t\in \nats$, let $\mathcal{B}(u,v,t)$ denote the event that $v$ has not received $u$'s packet within the first $t$ rounds. We now derive a bound on the conditional probability of $\mathcal{B}(u,v,t)$ given $\mathcal{E}^1_{\delta} \cap \mathcal{E}^2_{\delta}$, when Algorithm R1 is employed. 
	
	At the end of the first round, each node $w\in M_{uv} \subseteq N^{\rm out}_u$ receives $u$'s packet, amongst $d^{\rm in}_w$ packets in total. Under Algorithm R1, in each subsequent round, it broadcasts $u$'s packet with probability $1/d^{\rm in}_w$, independent of other nodes and rounds. For $B(u,v,t)$ to occur, no node $w\in M_{uv}$ should choose to transmit $u$'s packet in any of the rounds $2,\ldots,t$. Hence, we see that 
	\begin{align*}
	    \prob(B(u,v,t)|\mathcal{E}^1_{\delta}\cap \mathcal{E}^2_{\delta}) 
	    &= \E \Bigl[ \prod_{w\in M_{uv}} \Bigl( 1-\frac{1}{d^{\rm in}_w} \Bigr)^{t-1} \Bigm| \mathcal{E}^1_{\delta} \cap \mathcal{E}^2_{\delta} \Bigr] \\
	    &\leq \Bigl( 1- \frac{1}{(1+\delta)np} \Bigr)^{(1-\delta) (n-2) p^2(t-1)} \\
	    &\leq C\exp \Bigl( -\frac{(1-\delta)(n-2) pt}{(1+\delta) n} \Bigr), 
	\end{align*}
	where $C=\exp((1-\delta)p/(1+\delta))$ is a constant that does not depend on $n$.
	Setting $t=2(1+\epsilon)\log n/p$ and using the union bound, we obtain from the above that 
	\begin{equation}
	    \label{eq:relay_fail_bd}
	    \begin{aligned}
	    &\prob \Bigl( \bigcup_{u,v\in V} B(u,v,t) \Bigm| \mathcal{E}^1_{\delta}\cap \mathcal{E}^2_{\delta} \Bigr) \\ 
	    &\leq C\exp \Bigl( 2\log n -\frac{2(1+\epsilon)(1-\delta)(n-2) \log n }{(1+\delta) n} \Bigr).
	    \end{aligned}
	\end{equation}
	It is clear that, for given $\epsilon>0$, we can choose $\delta>0$ sufficiently small that the exponent on the RHS is a strictly negative multiple of $\log n$, for all $n$ sufficiently large. 
	
	Now, the probability that there are some $u,v\in V$ such that $v$ fails to receive $u$'s packet within $2(1+\epsilon)\log n/p$ rounds is bounded as follows:
	\begin{align*}
	    &\prob \Bigl( \bigcup_{u,v\in V} B\Bigl( u,v,\frac{2(1+\epsilon}{p}\log n \Bigr) \Bigr) \\
	    &\leq \prob \Bigl( \bigcup_{u,v\in V} B \Bigl( u,v,\frac{2(1+\epsilon)}{p} \log n \Bigr) \Bigm| \mathcal{E}^1_{\delta}\cap \mathcal{E}^2_{\delta} \Bigr) + \prob( (\mathcal{E}^1_{\delta})^c ) + \prob( (\mathcal{E}^2_{\delta})^c ).
	\end{align*} By (\ref{eq:relay_fail_bd}) and Lemma \ref{lem:degree_conc}, each of the
	summands on the RHS tends to zero as $n$ tends to infinity, for a suitable choice of
	$\delta>0$. Thus, we have proved the first claim of the theorem.
	
	The proof of the second claim is very similar. The only difference is that we use the bound that the probability of $w\in M_{uv}$ choosing to broadcast $u$'s packet in any round is at least $1/n$ under Algorithm R2 (as $w$ has at most $n$ packets).
\end{proof}

While Theorem \ref{thm:baseline} is stated under the assumption that the communication graph is
fixed over time, the careful reader will have noticed that the proof only relies on the properties
established in Lemma~\ref{lem:degree_conc}. More precisely, if we define $N^{\rm in}_v(t)$ and
$N^{\rm out}_v(t)$ to be the in and out neighbourhoods of vertex $v$ in the communication graph
$G_t=(V,E_t)$ at time step $t$, and define $M_{uv}(t)=N^{\rm out}_u(1)\cap N^{\rm in}_v(t)$, and if
the probability bounds in Lemma~\ref{lem:degree_conc} hold for the in and out-degrees in time step
1, and for the cardinalities of $M_{uv}(t)$ in each of the first $2(1+\epsilon)\log n/p^2$ time
steps, then the conclusion of Theorem~\ref{thm:baseline} still holds. An example of a random graph
model satisfying these conditions is one in which the edges evolve as independent On-Off Markov
chains, with stationary probability $p$ of being On. Simulation results for this model are presented
in Section~\ref{sec:sim}.

\section{Random linear network coding} \label{sec:coding}
In this section, we present a different approach to allcast based on network coding. As before, in the first round, each node broadcasts its own packet. In subsequent rounds, instead of forwarding packets, nodes compute random linear combinations, over some finite field, of packets received in the first round. They broadcast this linear combination, along with the coefficients used. When a node has received $n$ linearly independent vectors of coefficients (or $n-1$, excluding coefficients on its own packet), it can solve the resulting system of linear equations to decode all packets. Thus, the number of rounds required for allcast is related to the rank of random matrices of coefficients over finite fields. We now specify the algorithm precisely before analysing its performance.

For the purposes of describing the algorithm, each packet is considered to be a fixed length vector
over a finite field, $\gf_q$; linear combinations are computed in this vector space. The
performance of the algorithms we present is not sensitive to the choice of $q$, so we shall take
$q=2$ for definiteness. We use 0 and 1 to denote the additive and multiplicative identities in
$\gf_2$ as well as the real numbers 0 and 1; it will be clear from context which is
intended. The algorithms described below are parametrised by a constant $\beta>1$, which does not
depend on $n$, the system size. 

{\bf Algorithm RLNC($\beta$)}
\begin{enumerate} 
    \item In the first round, each node broadcasts its own packet. It also stores each packet received in the first round in a buffer. Let $N^{\rm in}_v$ denote the set of packets received by node $v$ in the first round, and let $d^{\rm in}_v=|N^{\rm in}_v|$.
    \item In each subsequent round, each node broadcasts a random linear combination over $\gf_2$ of packets it received in the first round, computed as follows. Let $X_{vw}(t) \in \gf_2$ denote the coefficient assigned by node $v$ to node $w$'s packet in round $t$. Then, $X_{vw}(t)$ are mutually independent random variables, and 
    \begin{equation*}
	\prob(X_{vw}(t)=1) = 
	    \begin{cases}
        (\beta \log d^{\rm in}_v)/d^{\rm in}_v, & w\in \nats^{\rm in_v}, \\
	0, & \mbox{otherwise.}
	\end{cases}
    \end{equation*}
		It transmits the corresponding coefficients, $X_{vw}(t)$, in the same packet.
\end{enumerate}

Strictly speaking, the overhead imposed by transmitting the coefficients means that packet sizes are not independent of the system size, $n$. If they were transmitted in the form of a coefficient vector, with one element for each node, the overhead would be $n$ bits. An alternative is to only transmit identifiers of nodes whose coefficient is 1. The number of such coefficients employed by the algorithm is random, but concentrated sharply around a multiple of $\log n$. The length of each identifier is $O(\log n)$. Thus, the overhead is $O(\log^2 n)$ bits, in probability and in expectation. Nevertheless, for practically relevant values of $n$, and practically relevant packet sizes, it is reasonable to assume that transmitting coefficients increases packet sizes by no more than a constant factor. For example, if packet sizes are 1Kb and $n=10^3$, then appending a vector of coefficients merely doubles the packet size.

Algorithm RLNC($\beta$) is fully distributed, and requires no knowledge about the system on the part of
individual agents or nodes. In particular, nodes do not need to know $n$ or $p$. As nodes only use packets received in the first round to compute the linear combinations broadcast in subsequent rounds, the algorithm relies for its correctness on the fact that the communication graph has diameter 2, whp. While one can modify the algorithm to use packets received in subsequent rounds in the linear combinations, the analysis of such a generalisation does not appear tractable using the approach in this paper.

Once each node has received a linearly independent set of $n-1$ random linear combinations, the messages can be recovered by solving the resulting linear system. We seek to bound the number of transmission rounds until this is possible for all nodes. Our main result is as follows.

\begin{theorem} \label{thm:coding}
Consider a sequence of systems indexed by the number of agents, $n$, and satisfying assumptions $A1$ and $A2$. Let $T^{\rm all}_n(RLNC(\beta))$ denote the number of rounds needed until all agents are able to decode all packets when they employ Algorithm RLNC($\beta$). Then, for any $\beta>8$, 
$$
\lim_{n\to \infty} \prob \bigl( T^{\rm all}_n (RLNC (\beta)) > \lceil \frac{1}{p} \rceil+2 \bigr) = 0.
$$
\end{theorem}

The theorem says that the number of rounds needed to disseminate all messages to all nodes remains bounded if network coding is employed. Equivalently, the total number of broadcasts required scales linearly in the number of nodes. Moreover, the number of rounds needed is at most two more than the lower bound of $\lceil 1/p \rceil$ required by any algorithm. This is in contrast to random relaying, which needs a number of rounds growing logarithmically in the system size, or a number of broadcasts growing as a multiple of $n\log n$.

The proof will proceed through a sequence of lemmas. In order to prove the theorem, we need to study the ranks of the random matrices of coefficients generated by the algorithm. We start by describing these matrices in detail.

Let $\matr{A}(t,v)$ denote the matrix of coefficients received by agent $v$ in rounds $2,3,\ldots,t+1$. This is a $td^{\rm in}_v\times n$ matrix. Each row corresponds to the coefficient vector chosen by one of the in-neighbours of $v$ in one of the rounds. The columns index the agents, i.e., the nodes in the graph. If the $j^{\rm th}$ element of a row is equal to 1, this means that the corresponding agent, in the corresponding round, included $j$'s packet in the linear combination that it broadcast.

We want to find the smallest value of $t$ such that $\matr{A}(t,v)$ has rank $n$ for every $v\in V$; then, it is guaranteed that all agents can decode all messages at the end of round $t+1$. The reason that we have excluded round 1 is only for ease of analysis. The coefficient vectors used by a given node from round 2 onwards are i.i.d., but have a different distribution in round 1. Ignoring round 1 yields an upper bound on the number of rounds needed, which differs by at most 1 from the true number.

Our analysis below will draw heavily on techniques used by Bl\"omer \emph{et al.}~\cite{apr19blomer97}, who study the ranks of sparse random matrices over finite fields. (The dense case was studied by Mukhopadhyay~\cite{mukhopadhyay84}.) The main difference is that, while they consider matrices with i.i.d. entries, there are dependencies in the matrices generated by our model, caused by the fact that they can only have 1s in positions corresponding to edges of the communication graph, which is unchanged from one round to the next. Hence, we need to adapt the analysis accordingly.

Recall that the kernel of an $m\times n$ $\gf_2$-valued matrix $\matr{A}$ is defined as 
\begin{equation} \label{eq:kernel_def}
    \mbox{ker}(\matr{A}) = \{ \boldx \in \gf_2^n: \matr{A} \boldx = \boldzero \}.
\end{equation}
We use $\boldzero$ to denote a vector all of whose elements are the additive identity in $\gf_2$, and refer to it as the zero vector.

The $m\times n$ matrix $\matr{A}$ has rank $n$ if and only if $\mbox{ker}(\matr{A}) =\{ \boldzero \}$. If $\matr{A}$ is a random matrix, then we can obtain a lower bound on the probability that it has rank $n$ by upper bounding the probability that there is a non-zero vector in the kernel; such an upper bound can be computed by taking the union bound on $\prob(\boldx \notin \mbox{ker}(\matr{A}))$ over all non-zero $\boldx \in \gf_2^n$, which is a finite set. We now apply this to the matrices $\matr{A}(t,v)$.

Fix $i,v\in V$. Let $\matr{R}_i(t,v)$ denote the $t\times n$ sub-matrix of $\matr{A}(t,v)$ composed of rows corresponding to broadcasts by node $i\in N^{\rm in}_v$. If $j\notin N^{\rm in}_i$, then column $j$ of $\matr{R}_i(t,v)$ has to be zero, since $i$ cannot use $j$'s packet as part of the linear combination it transmits. If $j\in N^{\rm in}_i$, then, by Algorithm RLNC($\beta$), the elements of column $j$ are independent $Bern(\pi_i)$ random variables, where $\pi_i= \beta \log (d^{\rm in}_i)/d^{\rm in}_i$. Finally, the sub-matrices $\matr{R}_i(t,v)$ of $\matr{A}(t,v)$ corresponding to distinct $i$ are mutually independent. Thus, we have fully specified the probability law of the random matrix $\matr{A}(t,v)$, given the communication graph. We now study its rank.

\begin{lemma}\label{lem:submat_kernel}
Fix $m,n\in \nats$, $p\in (0,1)$ and $\pi \in (0,1/2)$. Let $\matr{R} \in \gf_2^{m\times n}$ be a random matrix with mutually independent columns, constructed as follows: each column of $\matr{R}$ is zero with probability $1-p$; with the residual probability $p$, it is composed of i.i.d. $Bern(\pi)$ entries.

Fix $k\in \nats$, $1\leq k\leq n$, and a vector $\boldx_k \in \gf_2^n$ consisting of $k$ ones and $n-k$ zeros. Then, 
$$
\prob(\matr{R}\boldx_k = \boldzero) = \sum_{s=0}^k \binom{k}{s}p^s (1-p)^{k-s} \Bigl( \frac{1+(1-2\pi)^s}{2} \Bigr)^m.
$$
\end{lemma}

\begin{proof}
By symmetry, the probability of the event of interest does not depend on which $k$ elements of $\boldx_k$ are non-zero.
The event $\matr{R}\boldx_k=\boldzero$ is the event that the sum of the corresponding $k$ columns of $\matr{R}$ is zero. 

Now, each of these columns is zero with probability $1-p$. Let $k-S$ denote the random number of such columns, and note that $S$ has a $Bin(k,p)$ distribution. Then, $\matr{R}\boldx_k=\boldzero$ if and only if the sum of the $S$ columns of $\matr{R}$ which are not forced to be zero is the zero vector. By conditioning on the possible values of $S$, and using Lemma~\ref{lem:iid_rank} below, we thus obtain that 
\begin{align*}
    \prob(\matr{R}\boldx_k=\boldzero) 
    &= \sum_{s=0}^k \prob(S=s) \prob(\matr{R}\boldx_k=\boldzero | S=s) \\
    &= \sum_{s=0}^k \binom{k}{s}p^s (1-p)^{k-s} \Bigl( \frac{1+(1-2\pi)^s}{2} \Bigr)^m
\end{align*}
This establishes the claim of the lemma.
\end{proof}

\begin{lemma} \label{lem:iid_rank}
Let $\pi\in (0,1/2)$, and let $\matr{B}\in \gf_2^{m\times s}$ be a random $m\times s$ matrix whose elements are i.i.d. $Bern(\pi)$ $\gf_2$-valued random variables. Let $\boldsymbol{e}\in \gf_2^s$ denote the all-one column vector. Then, 
$$
\prob(\matr{B}\boldsymbol{e}=\boldzero) = P(s,\pi)^m, 
\mbox{ where } P(s,\pi) = \frac{1+(1-2\pi)^s}{2}.
$$
\end{lemma}

\begin{proof}
The elements of the vector $\matr{B}\boldsymbol{e}$ are mutually independent $\gf_2$-valued random variables. Thus, it suffices to show that each element is zero with probability $P(s,\pi)$. 

Each element of $\matr{B}\boldsymbol{e}$ is the sum, in $\gf_2$, of $s$ independent $Bern(\pi)$ random variables. Hence, we obtain the recursion 
$$
P(0,\pi)=1, \quad P(s,\pi) = (1-\pi)P(s-1,\pi)+\pi(1-P(s-1,\pi)).
$$
It is readily verified that the expression given in the statement of the lemma solves this recursion.
\end{proof}

Lemma~\ref{lem:submat_kernel} gives an exact expression for the probability that a vector with $k$ ones is in the kernel of a random matrix $\matr{R}$. We now use this to derive inequalities that are more convenient to work with.

\begin{lemma} \label{lem:submat_kernel_ineq}
Let $m, p,\pi,\matr{R}$ be as in the statement of Lemma~\ref{lem:submat_kernel}, and let $\boldx_k$ denote a $\gf_2$-valued vector with exactly $k$ ones. Then, 
$$
\prob(\matr{R}\boldx_k=\boldzero) \leq 2^{-m} \Bigl[ 1 + \exp(-kp\pi) + 2 \exp \Bigl(-\frac{k}{m}\kl \Bigl( \frac{p}{2};p \Bigr) \Bigr) \Bigr]^m.
$$
In addition,
$$
\prob(\matr{R}\boldx_k=\boldzero) \leq e^{-kmp\pi/4}, 
$$ 
for all $k\leq k^* = \max \bigl\{ k: (1-2\pi)^k \geq \sfrac{1}{2} \mbox{ and } \bigl( 1-\frac{\pi k}{2} \bigr)^m \geq \sfrac{1}{2} \bigr\}.
$
\end{lemma}

\begin{proof}
We obtain from Lemma~\ref{lem:submat_kernel} and the fact that $(1-2\pi)^s$ is a decreasing function of $s$ that 
\begin{equation}
    \label{eq:kernelprob_ineq1}
    \begin{aligned}
    \prob(\matr{R}\boldx_k=\boldzero) \leq 
    &\sum_{s=0}^{\lceil kp/2 \rceil-1} \binom{k}{s}p^s(1-p)^{k-s} \\
    & + \Bigl( \frac{1 + (1-2\pi)^{\lceil kp/2 \rceil}}{2} \Bigr)^m \sum_{s=\lceil kp/2 \rceil}^k \binom{k}{s}p^s(1-p)^{k-s} \\
    \leq & \prob \Bigl( Bin(k,p)\leq \frac{kp}{2} \Bigr) + \Bigl( \frac{1 + (1-2\pi)^{kp/2}}{2} \Bigr)^m.
    \end{aligned}
\end{equation}
The first term in the above sum is bounded by $\exp(-k \kl(p/2;p))$ by Lemma~\ref{lem:binom_ldp}, while $(1-2\pi)^{kp/2} \leq \exp(-kp\pi)$ since $e^{-x}\geq 1-x$. Hence, we obtain from (\ref{eq:kernelprob_ineq1}) that 
\begin{align*}
    \prob(\matr{R}\boldx_k=\boldzero) &\leq \Bigl[ \exp \Bigl(-\frac{k}{m}\kl \bigl( \frac{p}{2};p \bigr) \Bigr) \Bigr]^m + \Bigl( \frac{1 + \exp(-kp\pi)}{2} \Bigr)^m \\
    &\leq 2^{-m} \Bigl( 1 + \exp(-kp\pi) + 2 \exp \Bigl(-\frac{k}{m}\kl \bigl( \frac{p}{2};p \bigr) \Bigr)^m.
\end{align*}
The last inequality follows from the fact that $(a+b)^m \geq a^m+b^m$ for any $a,b \geq 0$ and any $m\in \nats$.

Thus, we have established the first claim of the lemma.
The proof of the second claim will use the inequality
\begin{equation} \label{eq:log_ineq}
(1-2x)^s \leq 1-xs, \quad \forall \; x\in (0,1/2), s\geq 0: (1-2x)^s \geq \frac{1}{2}.
\end{equation}
This is readily verified by noting that for any fixed $x\in (0,1/2)$, the function $g(s)=(1-2x)^s-1+xs$ satisfies $g(0)=0$ and $g'(s)< 0$ for all $s\in (0,\frac{\log 2}{-\log(1-2x)}$.

Hence, $(1-2\pi)^s \leq 1-\pi s$ for all $s\leq k\leq k^*$, and we obtain from Lemma~\ref{lem:submat_kernel} that 
\begin{align*}
    \prob(\matr{R}\boldx_k=\boldzero) 
    &\leq \sum_{s=0}^k \binom{k}{s}p^s (1-p)^{k-s} \Bigl( 1-\frac{\pi s}{2} \Bigr)^m \\
    &\leq \sum_{s=0}^k \binom{k}{s}p^s (1-p)^{k-s} \Bigl(1-\frac{m\pi}{4}s \Bigr) \\
    &= 1-\frac{kmp\pi}{4} \leq e^{-kmp\pi/4}.
\end{align*}
The second inequality is obtained by invoking (\ref{eq:log_ineq}) again, while the last inequality uses the elementary inequality $e^{-x}\geq 1-x$. This completes the proof of the lemma.
\end{proof}

\begin{lemma} \label{lem:kernel_bds}
Consider a sequence of directed \ER communication graphs $G(n,p)$ indexed by $n$, with given $p>0$. Fix $\delta>0$ and suppose that for each $n$, $G(n,p)$ belongs to the set $\mathcal{E}^1_{\delta}$ defined in Lemma~\ref{lem:degree_conc}.

Fix $v\in V$ and set $t=\lceil \frac{1}{p} \rceil +1$. Let $\matr{A}(t,v)$ denote the matrix of coefficients received by vertex $v$ in rounds $2,3,\ldots,t+1$ when Algorithm RLNC($\beta$) is employed. Let $\boldx_k \in \gf_2^n$ denote a vector with exactly $k$ ones.

Suppose $\beta>8$, and set $\gamma=1+(\beta/8)>2$. Let $\alpha>0$ be sufficiently small that 
$$
\Bigl( 1-\frac{\alpha \beta}{2p(1-\delta)} \Bigr)^t > \frac{1}{2} \mbox{ and } \exp \Bigl( -\frac{\alpha \beta}{ 2(1-\delta)} \Bigr) >\frac{1}{2}.
$$
Then, for all $n$ sufficiently large, we have the following:
$$
\prob(\matr{A}(t,v)\boldx_k=\boldzero) \leq \begin{cases}
  n^{-\gamma k}, & 1\leq k\leq \lfloor \frac{\alpha n}{\log n} \rfloor, \\
  \left( \frac{1+\exp(-\beta k\frac{\log n}{n})}{2} \right)^{(1-\delta)(1+p)n}, & \lceil \frac{\alpha n}{\log n} \rceil \leq k \leq n.
\end{cases}
$$
\end{lemma}

\begin{proof}
Let $\matr{R}_u(t,v)$ denote the matrix of coefficients received by $v$ from $u\in N^{\rm in}_v$ in time slots $2,\ldots,t+1$. Then, 
$$
\matr{A}(t,v) = \begin{pmatrix}
\matr{R}_1(t,v)^T & | & \ldots & | & \matr{R}_{d^{\rm in}_v}(t,v)^T
\end{pmatrix}^T
$$
up to a permutation of the rows, and it is clear that $\matr{A}(t,v) \boldx_k = \boldzero$ if and only if $\matr{R}_u(t,v) \boldx_k=\boldzero$ for every in-neighbour $u$ of $v$. Moreover, the matrices $\matr{R}_u(t,v)$, $u\in N^{\rm in}_v$, are i.i.d. 
Hence, 
$$
\prob(\matr{A}(t,v)\boldx_k=\boldzero) = \prod_{w\in N^{\rm in}_v} \prob(\matr{R}_w(t,v)\boldx_k = \boldzero) = \prob(\matr{R}_u(t,v) \boldx_k=\boldzero)^{d^{\rm in}_v},
$$
for an arbitrary $u\in N^{\rm in}_v$. As $d^{\rm in}_v \geq (1-\delta)np$ by the assumption that the event $\mathcal{E}^1_{\delta}$ occurs, it follows that  
\begin{equation} \label{eq:kernel_product_ineq}
    \prob(\matr{A}(t,v)\boldx_k=\boldzero) \leq \prob(\matr{R}_u(t,v)\boldx_k=\boldzero)^{ (1-\delta)np}.
\end{equation}

Next, the matrix $\matr{R}_u(t,v)$ has the same probability law as the matrix $\matr{R}$ in Lemmas~\ref{lem:submat_kernel} and \ref{lem:submat_kernel_ineq}, with $m=t$ and $\pi=\beta \log(d^{\rm in}_u)/d^{\rm in}_{u}$. Hence, the inequalities in Lemma~\ref{lem:submat_kernel_ineq} apply to $\prob(\matr{R}_u(t,v)\boldx_k=\boldzero)$. Noting that $d^{\rm in}_u \geq (1-\delta)np$ on the event $\mathcal{E}^1_{\delta}$, and that the bounds in Lemma~\ref{lem:submat_kernel_ineq} are a decreasing function of $\pi$, we infer that 
\begin{equation} \label{eq:kernelu_bd_allk}
\begin{aligned}
&\prob(\matr{R}_u(t,v)\boldx_k=\boldzero) \leq \\
& 2^{-t} \Bigl[ 1+ \exp\Bigl( -\beta k \frac{\log((1-\delta)np)}{ (1-\delta)n}\Bigr) + 2\exp\Bigl(-\frac{k}{t}\kl\Bigl( \frac{p}{2};p\Bigr) \Bigr) \Bigr]^t,
\end{aligned}
\end{equation}
for all $k\in \{1,\ldots,n\}$, and 
\begin{equation} \label{eq:kernelu_bd_smallk}
\begin{aligned}
&\prob(\matr{R}_u(t,v)\boldx_k=\boldzero) \leq \exp \Bigl( -\beta kt \frac{\log((1-\delta)np)}{4(1-\delta)n} \Bigr) \mbox{ for all } \\
&k: \Bigl( 1-2\beta \frac{ \log((1-\delta)np)}{(1-\delta)np} \Bigr)^k \geq \frac{1}{2} \mbox{ and } \Bigl( 1-\beta k \frac{\log((1-\delta)np)}{2(1-\delta)np} \Bigr)^t \geq \frac{1}{2}.
\end{aligned}
\end{equation}

Let $\alpha>0$ be as in the lemma. Then, 
\begin{equation} \label{eq:kcond1}
\Bigl( 1-2\beta\frac{\log(1-\delta)np}{(1-\delta)np} \Bigr)^{\frac{\alpha n}{\log n}} \sim \exp \Bigl( -\frac{2\alpha\beta}{(1-\delta)p} \Bigr) > \frac{1}{2}, 
\end{equation}
where the inequality holds by the assumptions of the lemma. We use the asymptotic notation $a_n\sim b_n$ to denote that $a_n/b_n$ tends to 1 as $n$ tends to infinity. Similarly,
\begin{equation}\label{eq:kcond2}
\Bigl( 1-\beta\frac{\alpha n}{\log n} \frac{\log( (1-\delta)np)}{2(1-\delta)np} \Bigr)^t \sim 
\Bigl( 1-\frac{\alpha \beta}{2p(1-\delta)} \Bigr)^t > \frac{1}{2}.
\end{equation}
It follows from (\ref{eq:kcond1}) and (\ref{eq:kcond2})  that, for all $n$ sufficiently large,
$$
    \Bigl( 1-2\beta\frac{\log(1-\delta)np}{(1-\delta)np} \Bigr)^k \geq \frac{1}{2} \mbox{ and } 
    \Bigl(1-\beta k\frac{\log((1-\delta)np)}{2 (1-\delta) np} \Bigr)^t \geq \frac{1}{2},
$$ 
if $k\leq {\frac{\alpha n}{\log n}}$.
Hence, we obtain from (\ref{eq:kernel_product_ineq}) and (\ref{eq:kernelu_bd_smallk}) that, for all $n$ sufficiently large and all $k\leq \frac{\alpha n}{\log n}$,
\begin{align*}
    \prob(\matr{A}(t,v)\boldx_k=\boldzero) &\leq
    \Bigl(  \exp \Bigl( -\beta kt \frac{ \log(1-\delta) np}{ 4(1-\delta)n } \Bigr) \Bigr)^{(1-\delta)np} \\
    &= \exp \Bigl( -\frac{\beta kpt \log((1-\delta)np)}{4} \Bigr) \\
    &\leq ((1-\delta)np)^{-\beta k/4} \leq n^{-\gamma k}.
\end{align*}
We have used the fact that $t=\lceil \frac{1}{p} \rceil+1$ to obtain the second inequality. The last inequality holds for all $n$ sufficiently large because $\beta$ was assumed to be larger than 8 and so $\gamma=1+(\beta/8)<\beta/4$.

We have thus established the first claim of the lemma.
Next, we turn to the second claim.
Substituting (\ref{eq:kernelu_bd_allk}) in (\ref{eq:kernel_product_ineq}) yields the inequality 
\begin{align*}
&\prob(\matr{A}(t,v)\boldx_k=\boldzero) \\
&\leq \Bigl[ \frac{ 1+\exp \bigl( -\beta k \frac{ \log((1-\delta) np)}{ (1-\delta)n}\bigr) + 2\exp \bigl(-\frac{k}{t}\kl\bigl( \frac{p}{2};pb\bigr) \bigr)}{2} \Bigr]^{(1-\delta)npt} \\
& \leq \Bigl[ \frac{1 + \exp \bigl( -\beta k \frac{\log((1-\delta)np)}{ (1-\delta)n}\bigr)}{2} \Bigr]^{(1-\delta)npt} \Bigl[ 1 + 2\exp\Bigl(-\frac{k}{t}\kl\Bigl( \frac{p}{2}; p\Bigr) \Bigr) \Bigr]^{(1-\delta)npt} \\
& \leq \Bigl[ \frac{1 + \exp \bigl( -\beta k \frac{\log((1-\delta)np)}{ (1-\delta)n} \bigr)}{2} \Bigr]^{(1-\delta)npt} \exp \Bigl[ 2(1-\delta)npt \exp \Bigl( -\frac{k}{t} \kl \Bigl( \frac{p}{2};p\Bigr) \Bigr) \Bigr].
\end{align*}
Since $t=\lceil \frac{1}{p} \rceil+1$ and $\kl(p/2;p)$ are fixed positive constants, it is easy to see that 
$$
npt \exp \Bigl( -\frac{k}{t} \kl \Bigl( \frac{p}{2};p\Bigr) \Bigr) = o(1) \quad \forall \; k\geq \frac{\alpha n}{\log n}.
$$
Furthermore, $npt \geq n(1+p)$. Hence, we obtain from the above that 
$$
\prob(\matr{A}(t,v)\boldx_k=\boldzero) 
\leq [1+o(1)] \Bigl[ \frac{1 + \exp\bigl( -\beta k \frac{\log((1-\delta)np)}{ (1-\delta)n} \bigr)}{2} \Bigr]^{(1-\delta)(1+p)n}.
$$
Since $\frac{\log((1-\delta)np}{(1-\delta)n}>\frac{\log n}{n}$ for all $n$ sufficiently large and fixed $\delta,p>0$, the second claim of the lemma is seen to hold.
\end{proof}

We are now ready to complete the proof of the main result about the number of rounds required by the $RLNC(\beta)$ algorithm.

\begin{proof}[Proof of Theorem~\ref{thm:coding}]
Observe that Algorithm $RLNC(\beta)$ fails to complete after $t=\lceil 1/p \rceil+2$ rounds only if there is some vertex $v$ for the which the matrix $\matr{A}(t,v)$ has a non-zero vector in its kernel, for $t=\lceil 1/p \rceil+1$. Here, $\matr{A}(t,v)$ is the matrix of coefficients received by vertex $v$ in rounds $2,\ldots,t+1$. In other words, there is some vertex $v$, some $k\geq 1$, and some vector $\boldx_k \in \gf_2^n$ with exactly $k$ 1s such that $\boldx_k \in \mbox{ker}(\matr{A}(t,v))$. As there are $\binom{n}{k}$ such vectors in $\gf_2^n$, each of which is equally likely to belong to $\mbox{ker}(\matr{A}(t,v))$, we obtain by the union bound that 
\begin{equation}
    \label{eq:completion_prob_bd}
    \begin{aligned}
    &\prob \Bigl( T^{\rm all}_n(RLNC(\beta)) > \lceil \frac{1}{p} \rceil+2 \Bigm| \mathcal{E}^1_{\delta} \Bigr) \\
    &\leq \prob \bigl( \exists \, v\in V, \boldx \neq \boldzero: \matr{A}(t,v)\boldx=\boldzero \bigm| \mathcal{E}^1_{\delta} \bigr) \\
    &\leq \sum_{v\in V} \sum_{k=1}^n \binom{n}{k} \prob( \matr{A}(t,v) \boldx_k = \boldzero | \mathcal{E}^1_{\delta} ), \quad t = \lceil \frac{1}{p} \rceil+1.
    \end{aligned}
\end{equation}
The event $\mathcal{E}^1_{\delta}$ was defined in Lemma~\ref{lem:degree_conc}.

We will now use Lemma~\ref{lem:kernel_bds} to bound the terms in the sum in the last line above. The constant $\delta>0$ appearing in its statement may be chosen arbitrarily small. We assume henceforth that it is chosen such that $(1-\delta)(1+p) > 1$. 

Let $\alpha>0$ be as in the statement of Lemma~\ref{lem:kernel_bds}. Then, it follows from the lemma that 
\begin{equation}
    \label{eq:sumbound_to_alpha}
    \sum_{k=1}^{\lfloor \frac{\alpha n}{\log n} \rfloor} \binom{n}{k} \prob( \matr{A}(t,v)\boldx_k = \boldzero | \mathcal{E}^1_{\delta}) 
    \leq \sum_{k=1}^{\lfloor \frac{\alpha n}{\log n} \rfloor} \frac{n^k}{k!} n^{-\gamma k} \leq \exp \bigl(n^{1-\gamma}\bigr) = o\Bigl( \frac{1}{n} \Bigr), 
\end{equation}
since $\gamma=1+(\beta/8)>2$.

Next, observe from Lemma~\ref{lem:kernel_bds} that, for $n$ sufficiently large and $k\geq \frac{\alpha n}{\log n}$, we have 
\begin{equation} \label{eq:kernelbd_alphaplus}
\prob(\matr{A}(t,v)\boldx_k=\boldzero | \mathcal{E}^1_{\delta})\leq \Bigl( \frac{1+\exp(-\alpha \beta)}{2} \Bigr)^{(1-\delta)(1+p)n}.
\end{equation}
Now, define 
$$
\theta = \inf \bigl\{ q\geq 0: \kl( q;\sfrac{1}{2} ) \leq \log \bigl( 1+e^{-\alpha \beta} \bigr) \bigr\}.
$$
The set over which the infimum is taken is non-empty, since $\kl(q;\sfrac{1}{2})$ decreases from $\log 2$ at $q=0$ to 0 at $q=\sfrac{1}{2}$. Hence, $\theta \in (0,\sfrac{1}{2})$, and $\kl(\theta;\sfrac{1}{2})=\log(1+\exp(-\alpha \beta))$ by the continuity of $\kl(\cdot;\sfrac{1}{2})$. Thus, it follows from (\ref{eq:kernelbd_alphaplus}) that 
\begin{align*}
    &\sum_{k=\lceil \frac{\alpha n}{\log n} \rceil}^{\lfloor \theta n \rfloor} \binom{n}{k} \prob(\matr{A}(t,v) \boldx_k=\boldzero | \mathcal{E}^1_{\delta}) \\
    &\leq \Bigl( \frac{1+\exp(-\alpha \beta)}{2} \Bigr)^{[(1-\delta)(1+p)-1]n} (1+\exp(-\alpha \beta))^n \sum_{k=\lceil \frac{\alpha n}{\log n} \rceil}^{\lfloor \theta n \rfloor} \binom{n}{k} 2^{-n} \\
    &\leq \Bigl( \frac{1+\exp(-\alpha \beta)}{2} \Bigr)^{[(1-\delta)(1+p)-1]n} (1+\exp(-\alpha \beta))^n \prob \bigl( Bin(n,\sfrac{1}{2})\leq \theta n \bigr) \\
    &\leq \Bigl( \frac{1+\exp(-\alpha \beta)}{2} \Bigr)^{[(1-\delta)(1+p)-1]n} (1+\exp(-\alpha \beta))^n e^{-n\kl(\theta,\sfrac{1}{2})}.
\end{align*}
We have used Lemma~\ref{lem:binom_ldp} to obtain the last inequality. Substituting $\kl(\theta;\sfrac{1}{2}) = \log(1+\exp(-\alpha \beta))$, we get 
\begin{equation}
    \label{eq:sumbound_to_theta}
    \begin{aligned}
    \sum_{k=\lceil \frac{\alpha n}{\log n} \rceil}^{\lfloor \theta n \rfloor} \binom{n}{k} \prob(\matr{A}(t,v) \boldx_k=\boldzero | \mathcal{E}^1_{\delta}) 
    &\leq \Bigl( \frac{1+\exp(-\alpha \beta)}{2} \Bigr)^{[(1-\delta)(1+p)-1]n} \\
    &= o\Bigl( \frac{1}{n} \Bigr),
    \end{aligned}
\end{equation}
since $\delta$ is chosen so that $(1-\delta)(1+p)>1$.

Next, observe from Lemma~\ref{lem:kernel_bds} that, for $n$ sufficiently large and $k\geq \theta n$, we have 
$$
\prob(\matr{A}(t,v)\boldx_k=\boldzero | \mathcal{E}^1_{\delta}) \leq \Bigl( \frac{1+n^{-\theta \beta}}{2} \Bigr)^{(1-\delta)(1+p)n}.
$$
Hence,
\begin{align*}
    &\sum_{k=\lceil \theta n \rceil}^{\lfloor n/8 \rfloor} \binom{n}{k} \prob(\matr{A}(t,v) \boldx_k=\boldzero | \mathcal{E}^1_{\delta}) \\ 
    &\leq \Bigl( \frac{1+n^{-\theta \beta}}{2} \Bigr)^{[(1-\delta)(1+p)-1]n} \bigl( 1+n^{-\theta \beta} \bigr)^n \sum_{k=\lceil \theta n \rceil}^{\lfloor n/8 \rfloor} \binom{n}{k} 2^{-n} \\
    &\leq \Bigl( \frac{1+n^{-\theta \beta}}{2} \Bigr)^{[(1-\delta)(1+p)-1]n} \exp \bigl( n^{1-\theta \beta} \bigr) \prob \Bigl( Bin(n,\sfrac{1}{2})\leq \frac{n}{8} \Bigr) \\
    &\leq \Bigl( \frac{1+n^{-\theta \beta}}{2} \Bigr)^{[(1-\delta)(1+p)-1]n} \exp \Bigl( n^{1-\theta \beta} - n\kl\Bigl( \frac{1}{8}; \frac{1}{2} \Bigr) \Bigr).
\end{align*}
We have used Lemma~\ref{lem:binom_ldp} to obtain the last inequality. Notice that the bound holds trivially if $\theta>1/8$, since an empty sum is zero by convention. Now, $(1-\delta)(1+p)>1$ by the choice of $\delta$, while $n^{1-\delta \theta}<n\kl(1/8;1/2)$ for all $n$ sufficiently large. Hence, it follows that 
\begin{equation}
    \label{eq:sumbound_to_eighth}
    \sum_{k=\lceil \theta n \rceil}^{\lfloor n/8 \rfloor} \binom{n}{k} \prob(\matr{A}(t,v) \boldx_k=\boldzero | \mathcal{E}^1_{\delta}) = o\Bigl( \frac{1}{n} \Bigr).
\end{equation}

Finally, for $k\geq n/8$ and sufficiently large $n$, we see from Lemma~\ref{lem:kernel_bds} that
\begin{align*}
\prob(\matr{A}(t,v)\boldx_k=\boldzero | \mathcal{E}^1_{\delta}) &\leq \Bigl(\frac{1+n^{-\beta/8}}{2}\Bigr)^{(1-\delta)(1+p)n} \\ 
&\leq 2^{-(1-\delta)(1+p)n}\exp \bigl( (1-\delta)(1+p)n^{1-(\beta/8)} \bigr) \\
&= 2^{-(1-\delta)(1+p)n}(1+o(1)),
\end{align*}
since $\beta>8$ by assumption.
Hence,
\begin{equation} \label{eq:sumbound_bulk}
\begin{aligned}
    \sum_{k=\lceil n/8 \rceil}^{n} \binom{n}{k} \prob(\matr{A}(t,v) \boldx_k=\boldzero | \mathcal{E}^1_{\delta}) 
    &\leq (1+o(1))2^{-[(1-\delta)(1+p)-1]n} \sum_{k=0}^n \binom{n}{k} 2^{-n} \\ 
    &= o\Bigl( \frac{1}{n} \Bigr),
\end{aligned}
\end{equation}
since $\sum_{k=0}^n \binom{n}{k}2^{-n}=1$, and $(1-\delta)(1+p)>1$.

Substituting equations (\ref{eq:sumbound_to_alpha})
(\ref{eq:sumbound_to_theta}), (\ref{eq:sumbound_to_eighth}) and (\ref{eq:sumbound_bulk}) in (\ref{eq:completion_prob_bd}), we get
$$
\prob \Bigl( T^{\rm all}_n(RLNC(\beta)) > \lceil \frac{1}{p} \rceil+2 \Bigm| \mathcal{E}^1_{\delta} \Bigr) \leq \sum_{v\in V} o\Bigl( \frac{1}{n} \Bigr) = o(1).
$$
In addition, $\prob \bigl( (\mathcal{E}^1_\delta)^c \bigr) = o(1)$ by Lemma~\ref{lem:degree_conc}. Combining these bounds with the inequality, 
\begin{align*}
&\prob \Bigl( T^{\rm all}_n(RLNC(\beta)) > \lceil \frac{1}{p} \rceil+2 \Bigr) \\
&\leq \prob \Bigl( T^{\rm all}_n(RLNC(\beta)) > \lceil \frac{1}{p} \rceil+2 \Bigm| \mathcal{E}^1_{\delta} \Bigr) + 
\prob \Bigl( \bigl( \mathcal{E}^1_{\delta} \bigr)^c \Bigr),
\end{align*}
yields the claim of the theorem.
\end{proof}

\section{Simulation results} \label{sec:sim}
The analytical results obtained in the last two sections are of an asymptotic character. Moreover, they only pertain to static networks. To complement these, we present Monte Carlo simulations for a range of network sizes, both in static and time-varying networks. In addition, we use simulations to explore the gaps between results that we can prove and which we conjecture, as described in more detail below. The simulations are written in CUDA C. Each result depicted in the figures below is based on 10,000 simulation runs, summarised in a box-and-whisker plot. The box shows the median and upper and lower quartiles from the 10,000 runs, while the whiskers show the full range of data observed.

The simulations of static networks demonstrate that the asymptotic results capture the performance of networks of a moderate size (100+ nodes), but may break down for small networks. We see that in the time-varying network models which we simulated, the algorithms perform significantly better than the theoretical bounds, lending support to our conjecture that static networks are the worst case. We now discuss the simulation results in detail. 

\begin{figure}[htbp]
	\centering
	\includegraphics[width=0.58\textwidth]{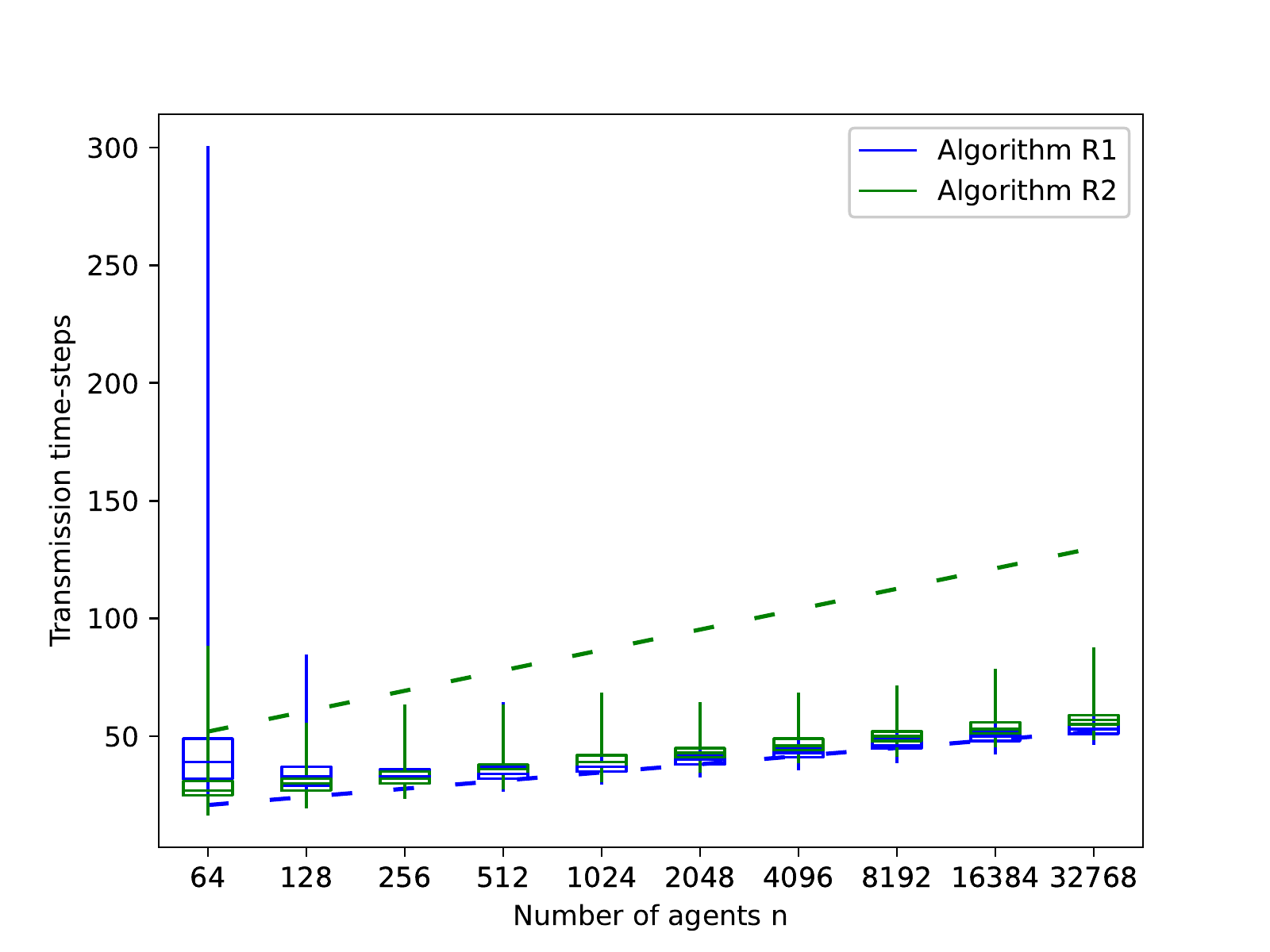}
	\caption{Number of rounds required for allcast by Algorithms R1 and R2 from
	Section~\ref{sec:fwd} on \ER random graphs $G(n,p)$. Plotted against $n$ on a semilog scale,
	for fixed  $p=0.4$. Box plots based on 10,000 replicates each. Dashed lines correspond to
	theoretical bounds of $2\log n/p$ (lower line) and $2\log n/p^2$ (upper line) from
	Theorem~\ref{thm:baseline}.
	}
	\label{fig:baseline_sims}
\end{figure}

\subsection{Static networks}

Figure~\ref{fig:baseline_sims} shows the performance of the two random relaying algorithms,
Algorithm R1 and R2, described in Section~\ref{sec:fwd}. The algorithms were simulated on 10,000
independent realisations of \ER random graphs $G(n,p)$, with $p=0.4$ and different values of $n$.
For each value of $n$ and $p$, the box plots show the median and upper and lower quartile of the
number of rounds needed for all nodes to receive all messages; the whiskers show the full range of
values seen in the simulations. The dashed lines in the figure depict the theoretical upper bounds
of $\frac{2\log n}{p}$ and $\frac{2\log n}{p^2}$ from Theorem~\ref{thm:baseline} on the number of
rounds required for allcast. As $n$ is displayed on a logarithmic scale on the $x$-axis, these
bounds appear as straight lines in the figure.

We observe from the figure that the box plots corresponding to Algorithms R1 and R2 sit virtually on
top of each other and are hard to distinguish. Thus, the simulations show that Algorithm R2 performs
as well as Algorithm R1, even though the bound we were able to prove in Theorem~\ref{thm:baseline}
on the number of rounds required was much weaker for Algorithm R2 than for Algorithm R1. The
simulations suggest that the number of rounds required by both algorithms are close to $2\log n/p$.
Finally, we notice that the whiskers are very wide for $n=64$, the smallest value of $n$ for which
results are depicted. This reflects that the asymptotic analysis becomes poorer for small values of
$n$. This is driven by variability in the \ER random graphs, which may contain nodes with very low
degrees (or even isolated nodes if we simulate enough instances for small enough $n$); these are
bottlenecks for allcast.

\begin{figure}[htbp]
	\centering 
	\includegraphics[width=0.58\textwidth]{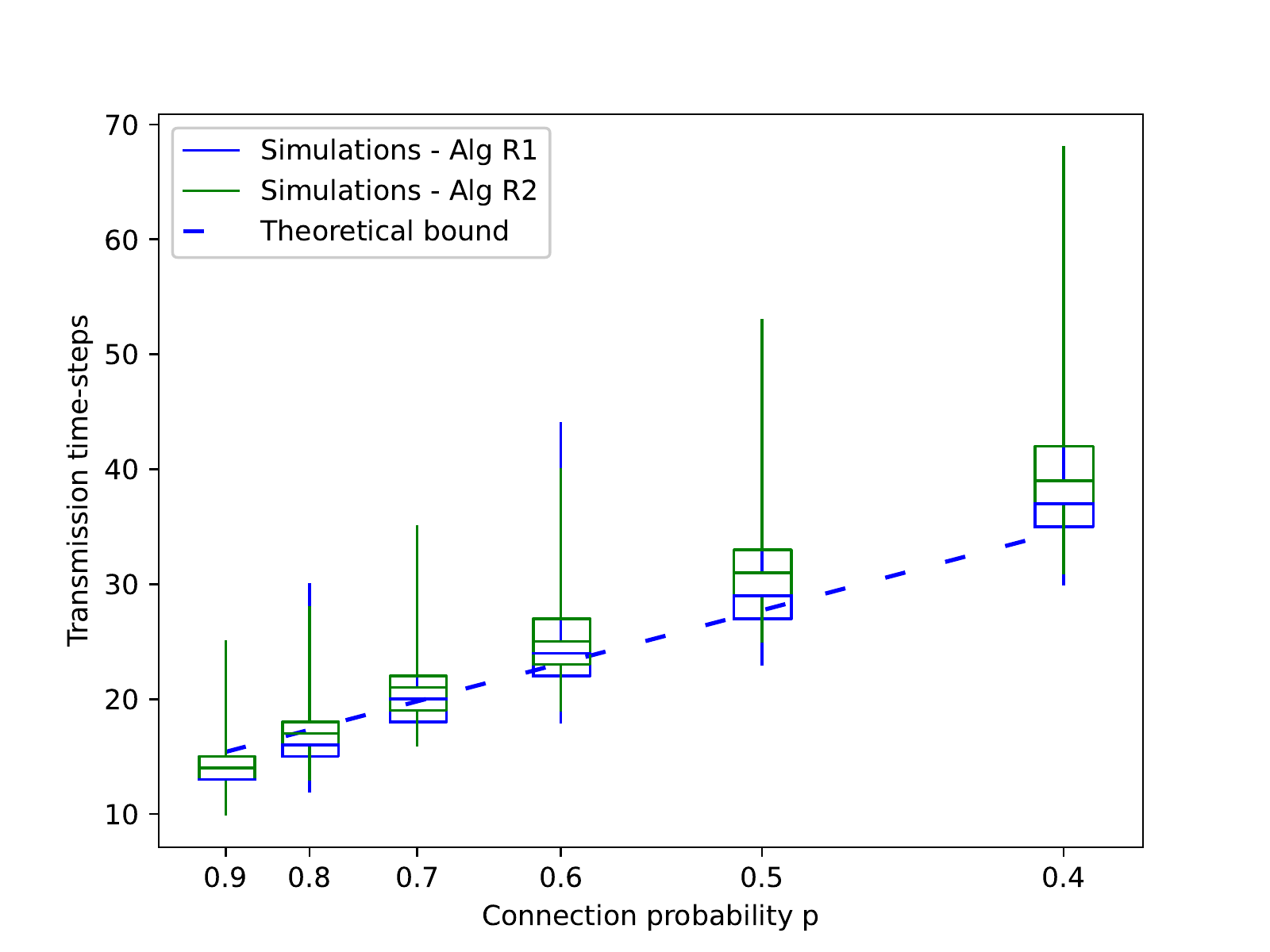} 
	\caption{Number of rounds required for allcast by Algorithms~R1 and~R2 from
	Section~\ref{sec:fwd} on \ER random graphs $G(n,p)$. Plotted against $1/p$ for fixed
	$n=1024$. Box plots based on 10,000 replicates each. The dashed line corresponds to the
	theoretical bound of $2\log n/p$ from Theorem~\ref{thm:baseline}.
	}
	\label{fig:baseline_sims2}
\end{figure}

While Figure~\ref{fig:baseline_sims} showed how the number of rounds required for allcast depends on $n$ for fixed $p$, we show how it depends on $p$ for fixed $n$ in  Figure~\ref{fig:baseline_sims2}. We fix $n=1024$ and run Algorithms R1 and R2 on 10,000 independent realisations of \ER random graphs $G(n,p)$ for different values of $p$. We plot the number of rounds required for allcast against $1/p$; the theoretical bound $2\log n/p$ from Theorem~\ref{thm:baseline} is the dashed straight line in the figure. It is clear from the figure that the relation between the number of rounds and $1/p$ is linear for both Algorithm R1 and R2, and that the number of rounds is very close to the theoretical bound of $2\log n/p$. We have not plotted the theoretical bound of $2\log n/p^2$ for Algorithm R2 as it is clear that this bound is not tight and does not capture the true performance of this algorithm. Finally, we observe that the number of rounds required becomes more variable as $p$ decreases. This shows that stochastic fluctuations become more important as the graph becomes sparser.

\begin{figure}[htbp]
	\centering
	\includegraphics[width=0.6\textwidth]{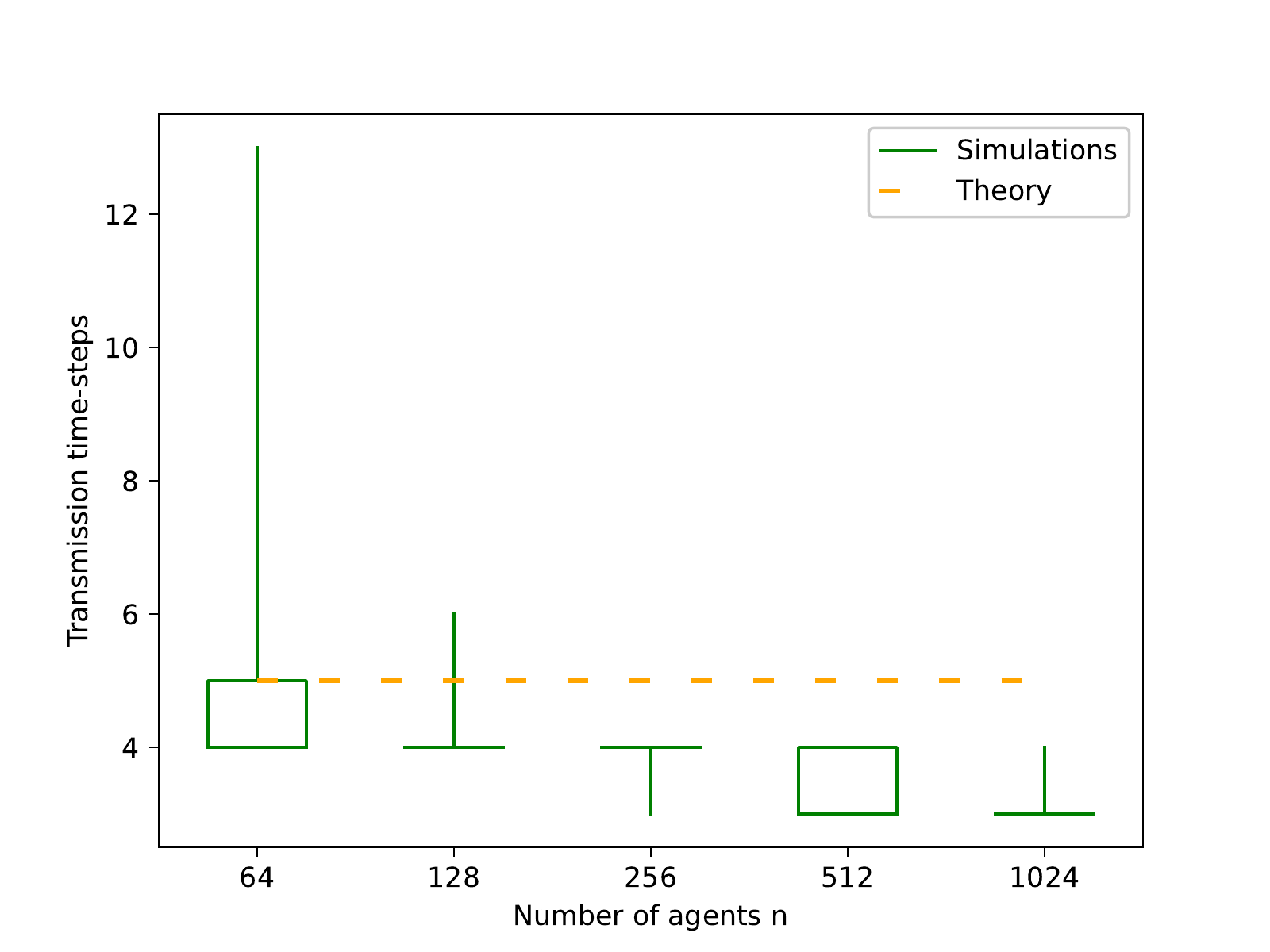}
	\caption{Number of transmission rounds required by Algorithm RLNC(8), on \ER random graphs
	$G(n,p)$ with $p=0.4$ and varying $n$. Box plots based on 10,000 replicates. The dashed
	horizontal line is the theoretical bound $\lceil 1/p \rceil+2 = 5$ from
	Theorem~\ref{thm:coding}.}
	\label{fig:trans_sims}
\end{figure}

We now turn to the performance of Algorithm $RLNC(\beta)$ from Section~\ref{sec:coding}, which
employs network coding.  Figure~\ref{fig:trans_sims} shows the number of rounds required to achieve
allcast on \ER random graphs $G(n,p)$ with $p=0.4$ and different values of $n$, when using
$RLNC(\beta)$ with $\beta=8$. The dashed horizontal line is the theoretical bound of $\lceil 1/p
\rceil+2 = 5$ rounds from Theorem~\ref{thm:coding}. The box plots, based on 10,000 replicates, are
close to the theoretical bound, and also show low variability. This demonstrates that not only does
network coding achieve much lower delay than forwarding on average, it is also more reliable.  The
relation between the number of rounds required and the edge probability $p$ of the random graphs is
depicted in Figure~\ref{fig:trans_sims2}, fixing the number of nodes at $n=256$ and the RLNC
sparsity parameter at $\beta=8$. The box plots are based on 10,000 replicates and the dashed line is
the theoretical bound of $\lceil 1/p \rceil+2$. The plots again show that the vast majority of
simulations are close to the theoretical bound, and the variability exhibited in the remainder is
not large. In fact, we observe in some of the plots that the median and upper and lower quartiles
all coincide, demonstrating how sharply the distribution is concentrated.

\begin{figure}[htbp]
	\centering
	\includegraphics[width=0.6\textwidth]{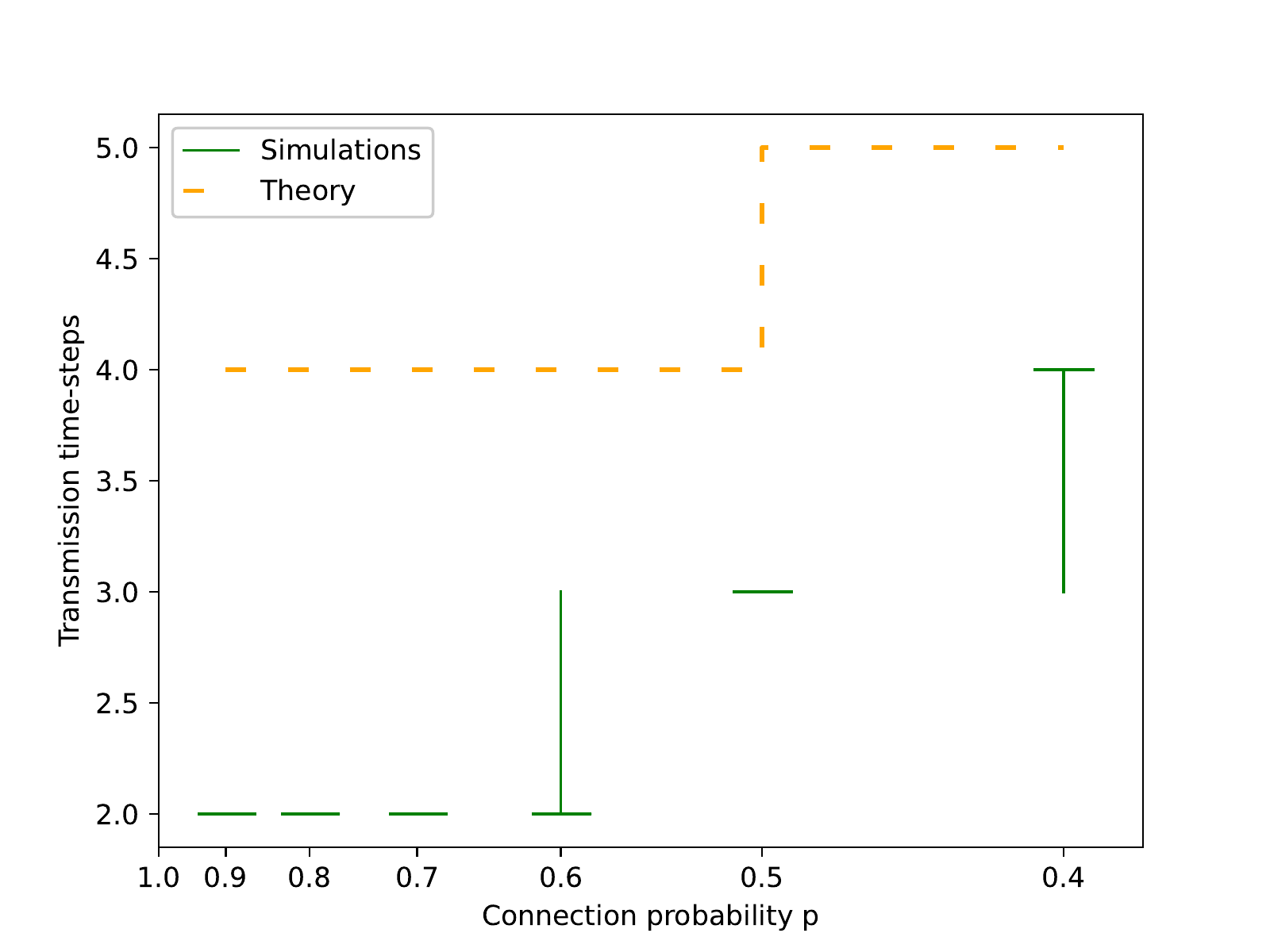}
	\caption{Number of transmission rounds required by Algorithm RLNC(8), on \ER random graphs
	$G(n,p)$ with $n=256$ and varying $p$. Box plots based on 10,000 replicates. The dashed line
	is the theoretical bound $\lceil 1/p \rceil+2 $ from Theorem~\ref{thm:coding}.}
	\label{fig:trans_sims2}
\end{figure}

Next, we investigate the affect of the sparsity parameter $\beta$ of the network code on the
algorithm's performance. Theorem~\ref{thm:coding} suggests that this constant should be taken to be
greater than $8$. However, our simulations show this to be unduly conservative, and that the
algorithm works for much smaller $\beta$. In Figure~\ref{fig:beta}, we compare the performance of
$RLNC(\beta)$ for different values of $\beta$, across a wide range of network sizes, $n$. The plots
show that, while decreasing $\beta$ worsens performance, the impact is small; decreasing $\beta$ from 8 to 4 has a negligible effect, and going from $\beta=4$ to $\beta=2$ typically increases the number of rounds required for allcast by 1 or 2.  The impact is
more pronounced in smaller networks, and manifests mainly in greatly increased variability in the
tail, but only a small increase in the typical number of round required. If we further decrease
$\beta$ to 1, then there is a marked increase in the number of rounds typically required; in small
networks, performance is severely degraded. This is consistent with the results
in~\cite{apr19blomer97}, which show that $\beta=1$ is the critical density of non-zero entries per
row required for random matrices over finite fields to have full rank. In summary, the simulations
suggest that, even though $\beta>8$ is assumed in Theorem~\ref{thm:coding}, we obtain performance
close to the bound in this theorem for all $\beta \geq 2$. On sufficiently large networks, even
$\beta \geq 1$ might be adequate.

\begin{figure}[htb!]
	\begin{subfigure}{0.5\textwidth}
		\centering
		\includegraphics[width=0.8\textwidth]{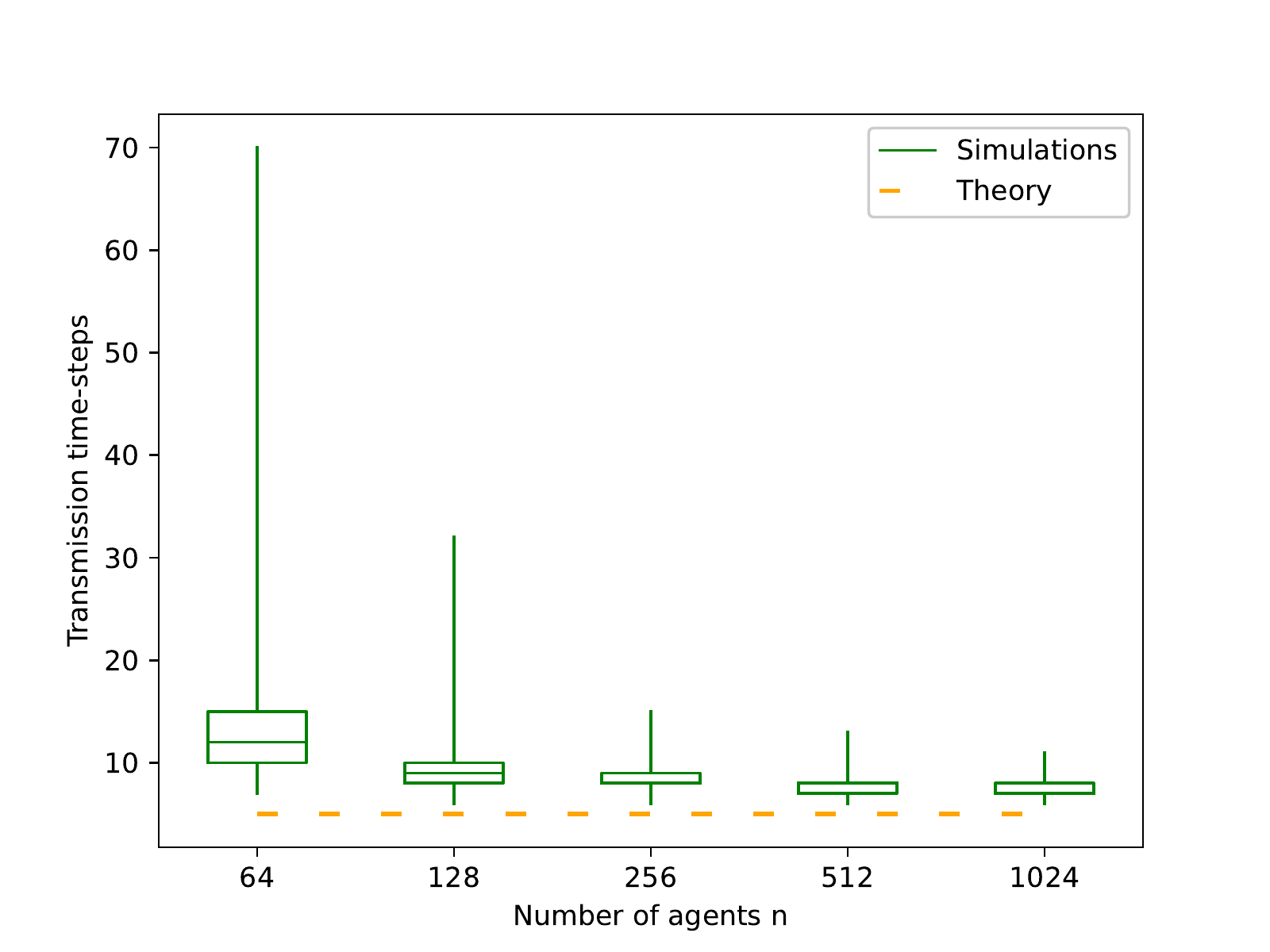}
		\caption{$\beta = 1$}
	\end{subfigure}
	\begin{subfigure}{0.5\textwidth}
		\centering
		\includegraphics[width=0.8\textwidth]{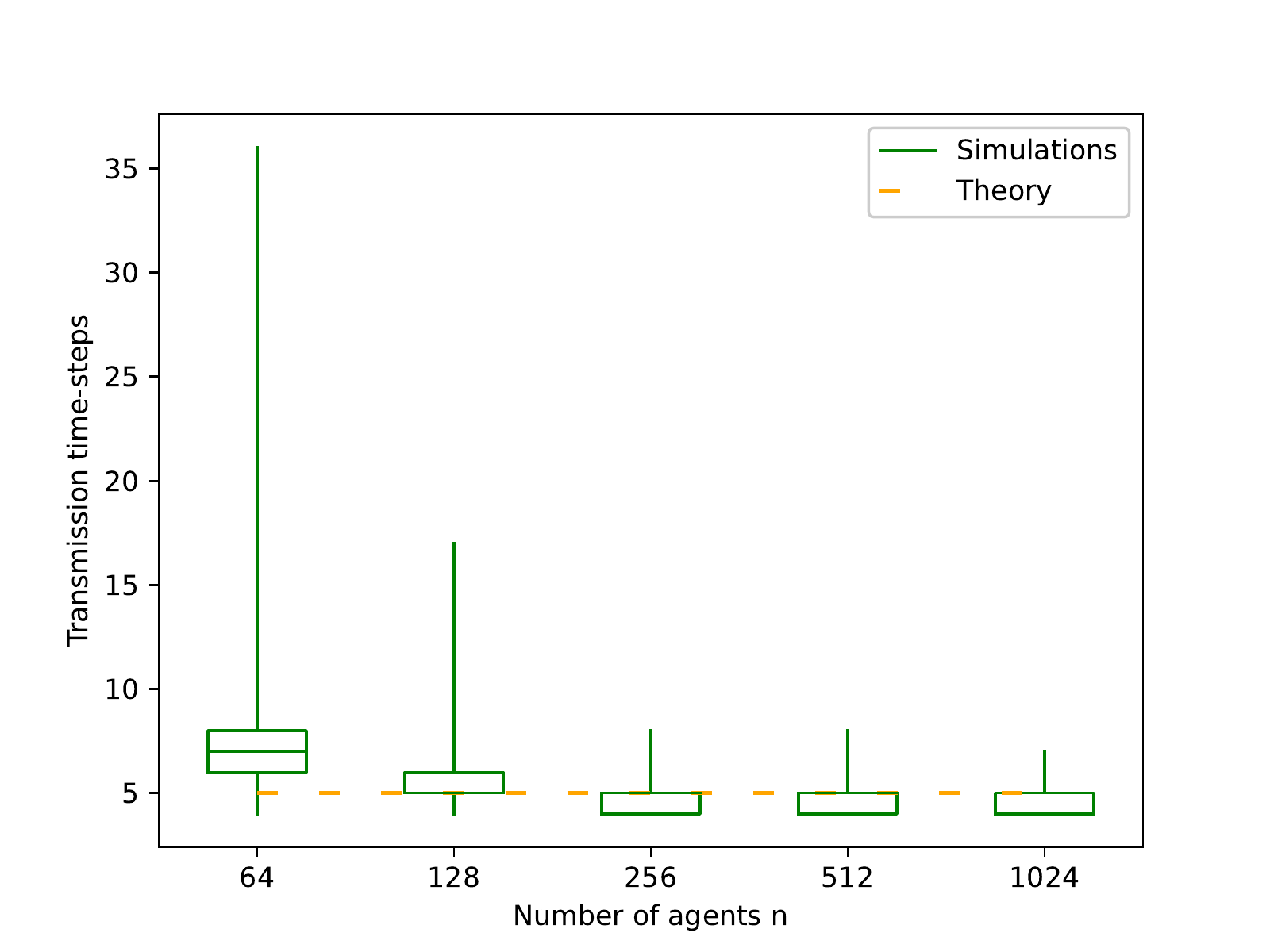}
		\caption{$\beta = 2$}
	\end{subfigure}
	\begin{subfigure}{0.5\textwidth}
		\centering
		\includegraphics[width=0.8\textwidth]{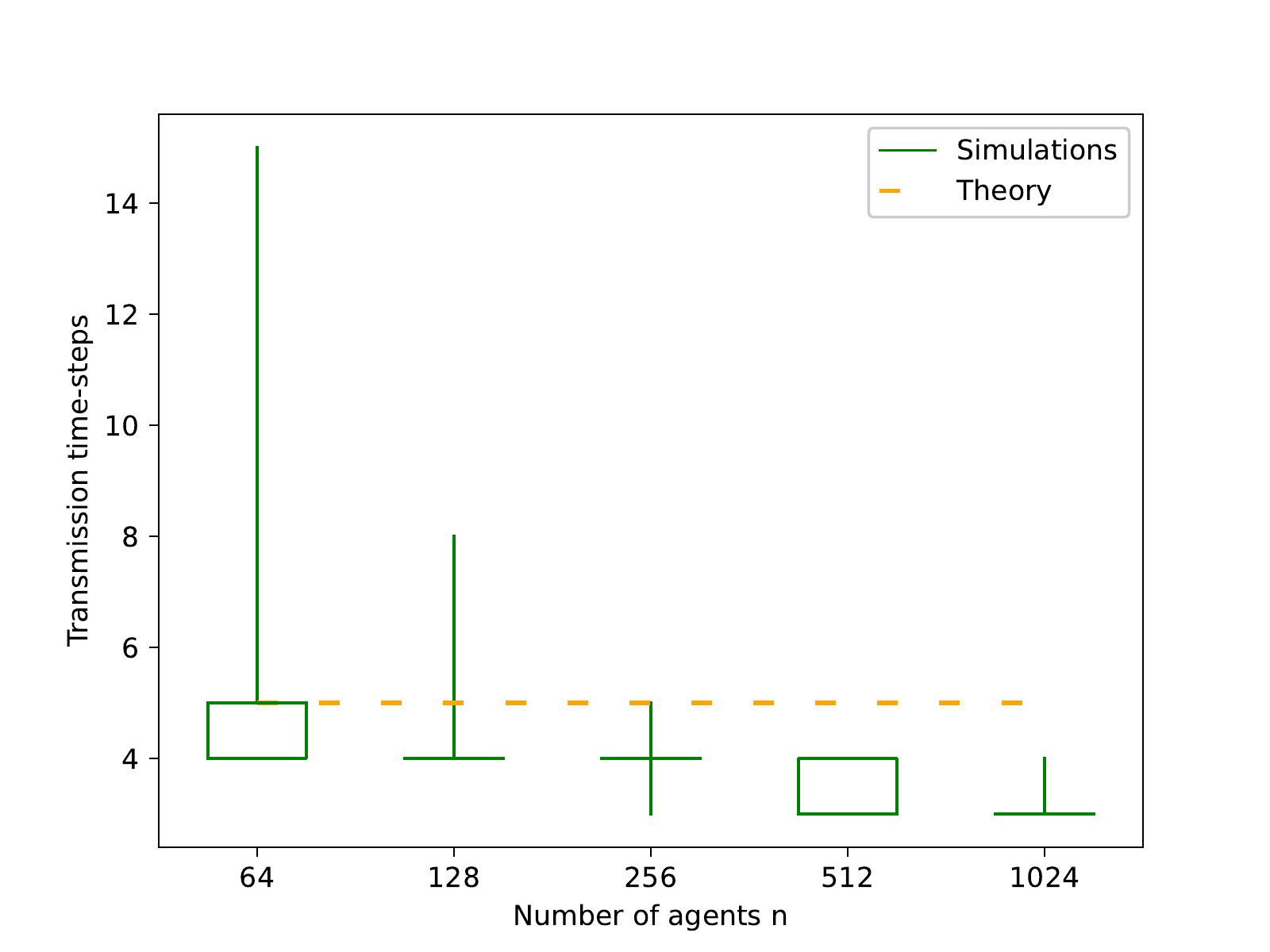}
		\caption{$\beta = 4$}
	\end{subfigure}
	\caption{Number of transmission rounds required by Algorithm $RLNC(\beta)$ to achieve allcast, plotted as a function of network size $n$, for $\beta=1,2,4$. Each boxplot is based on 10,000 \ER random graphs $G(n,p)$, with $p=0.4$. Dashed lines show the bound from Theorem~\ref{thm:coding}.}
	\label{fig:beta}
\end{figure}
\begin{figure}[p]
	\centering
	\includegraphics[width=0.6\textwidth]{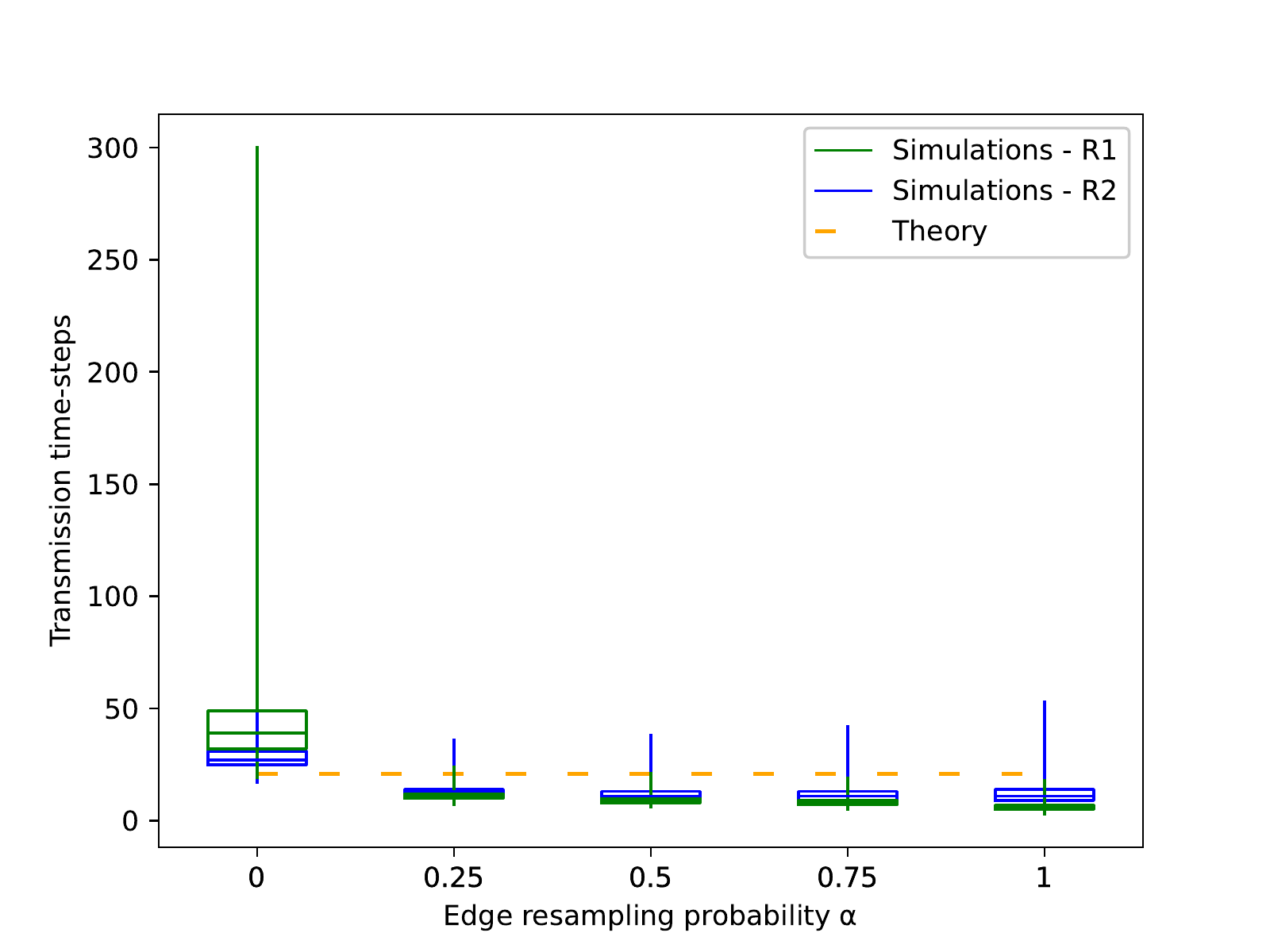}
	\caption{Number of rounds required by Algorithms~R1 and~R2 for allcast on
	evolving networks, plotted against the edge resampling probability, $\alpha$. Networks have
	$n=64$ nodes and edge probability $p=0.4$. Boxplots based on 10,000 replicates. The dashed
	line depicts the asymptotic bound $\frac{2}{p} \log(n) = 20.8$ rounds for static networks
	from Theorem~\ref{thm:baseline}.}
	\label{fig:alpha_baseline} 	
\end{figure}
\begin{figure}[p]
	\centering
	\includegraphics[width=0.6\textwidth]{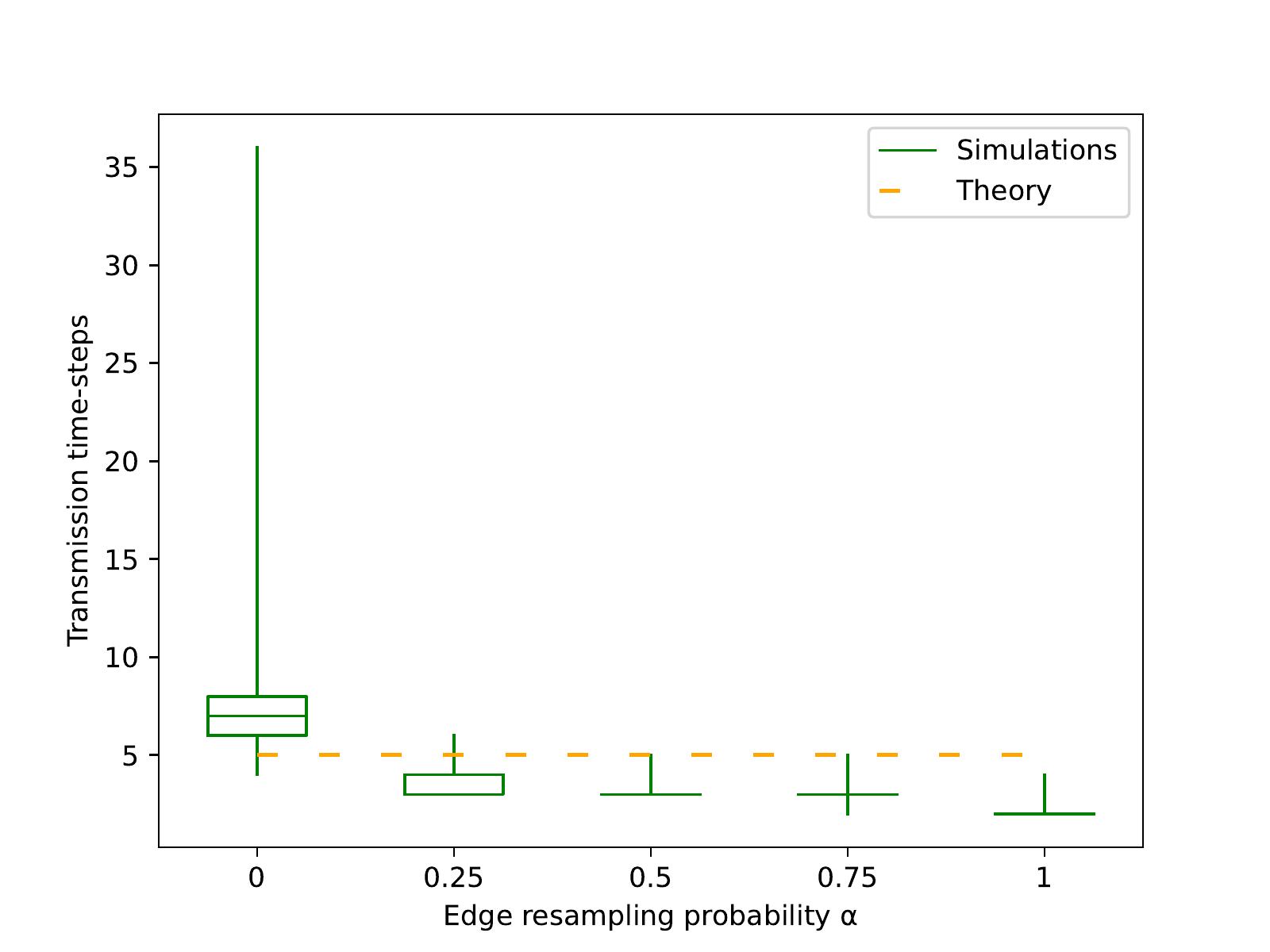}
	\caption{Number of rounds required by Algorithm RLNC(2) for allcast on evolving networks,
	plotted against the edge resampling probability, $\alpha$. Networks have $n=64$ nodes and
	edge probability $p=0.4$. Boxplots based on 10,000 replicates. The dashed line
	depicts the asymptotic bound $\lceil 1/p\rceil+2=5$ rounds for static networks from
	Theorem~\ref{thm:coding}.} \label{fig:alpha}
\end{figure}

\subsection{Time-varying networks}
The theoretical analysis presented in this paper is for static networks. This was justified on
grounds of timescale separation, i.e., the timescale on which the network topology changes is far
slower than the latencies acceptable for allcast. Nevertheless, there may be applications or
environments in which this assumption is violated, and so we explore the performance of our
algorithms in settings which relax this assumption.

We consider a network model in which the edges evolve as independent On-Off Markov chains. In each
time step $t$, and for each $(u,v)\in V\times V$, we retain the state of this ordered pair with
probability $1-\alpha$, where $\alpha\in [0,1]$ is a parameter of the model. Thus, with probability
$1-\alpha$, we take $(u,v)\in E_t$ if and only if $(u,v)\in E_{t-1}$. With the residual probability
$\alpha$, we resample the edge from a Bernoulli($p$) distribution, independent of everything else.
We could equivalently describe the Markov chain by stating that an edge changes from Off to On with
probability $\alpha p$ and from On to Off with probability $\alpha(1-p)$ in each time step. It is
easy to see that the stationary distribution for each edge is to be present with probability $p$,
and we initialise the edges independently with this distribution to obtain a stationary Markov
chain. Observe that we recover static networks if we take $\alpha=0$; if $\alpha=1$, the whole
network is resampled at each time step and there is no temporal correlation. 

Figures~\ref{fig:alpha_baseline}
and~\ref{fig:alpha} depict the results for Algorithms~R1, R2 and $RLNC(\beta)$, with $\beta=2$,
employed on time-varying networks with $n=64$ nodes and steady-state edge probability $p=0.4$. We
have plotted the number of rounds required for allcast by each algorithm against the parameter
$\alpha$, which measures how quickly the network changes over time.  Each boxplot is based on 10,000
replications, with each replication run until allcast completes over the evolving network.  We see
from the figures that for all algorithms, as $\alpha$ increases, the time required for allcast
decreases.  This provides some support for our conjecture that static networks are the worst case.

\section{Conclusion} \label{sec:concl}

The goal of this work was to develop lightweight distributed low-latency algorithms for the allcast problem, of disseminating one packet from each node in a connected communication graph to all other nodes. We presented two classes of algorithms, one based on random relaying and another based on random linear network coding. We considered a slotted time model in which each node broadcasts once in each time slot, and broadcasts are received error-free by its neighbours in the communication graph. We evaluated the performance of our algorithms on directed \ER random graphs $G(n,p)$, where $n$ denotes the number of nodes and $p$ the edge probability.  We presented analytical results in an asymptotic regime where $n$ tends to infinity with $p$ fixed; we showed that the time complexity (number of time slots required for allcast) of random relaying grows logarithmically in the number of nodes, whereas the time complexity of random coding remains bounded, irrespective of the number of nodes. The analytical results were complemented by Monte Carlo simulations for a wide range of system sizes. These showed that the asymptotic results provide a good approximation to the observed performance in systems of moderate size, ranging from 128 nodes upwards. They also showed that random coding substantially outperforms random relaying over the whole range of simulated parameters, by between a factor of 4 to 10. Finally, while our analytical results only pertain to static networks, the simulations showed that the proposed algorithms are robust to network dynamics; in fact, their performance improves if networks evolve over time.

A major limitation of the analytical results presented in this paper is that they only pertain to two-hop networks. The algorithms can be easily extended to general networks, but the analysis is greatly complicated by the dependencies induced. A new approach is needed to extend the analysis to general networks. Secondly, we did not study the numerical complexity of decoding the random linear codes proposed, or seek to develop efficient algorithms for this purpose; for example, we did not evaluate the convergence of belief propagation algorithms for our coding scheme. These are open problems for future research.

{\bf Acknowledgments} Work carried out while MAG was with the School of Mathematics and the CDT in
Communications, University of Bristol, United Kingdom. The authors would like to thank Steve Wales,
Roke Manor Research, for useful discussions.

\FloatBarrier

\bibliographystyle{IEEEtran} 
\bibliography{phd}
\end{document}